\documentclass[12pt]{article}
\usepackage{amsfonts}
\usepackage{fancyhdr}
\usepackage{comment}
\usepackage[a4paper, bottom=2cm, left=2cm, right=2cm, top = 1.5
cm]%
{geometry}
\usepackage{times}
\usepackage{amsmath}
\usepackage{physics}
\usepackage{changepage}
\usepackage{amssymb}
\usepackage{amsthm}
\usepackage{graphicx}%
\usepackage{multirow}
\usepackage{caption}
\usepackage{subcaption}
\usepackage{enumitem}
\usepackage{booktabs}
\usepackage{bm}
\usepackage{bbm}
\usepackage{mdframed}
\usepackage{newtxmath}
\usepackage{hyperref}
\usepackage{indentfirst}
\usepackage{csquotes} 
\usepackage[title]{appendix}

\DeclareMathOperator{\sgn}{sgn}
\newtheorem{theorem}{Theorem}
\newtheorem{corollary}{Corollary}
\newtheorem{lemma}{Lemma}
\newtheorem{assumption}{Assumption}

\theoremstyle{definition}
\newtheorem{definition}{Definition}

\usepackage{bibentry}
\usepackage{natbib}
\nobibliography*

\usepackage{algorithm,algorithmic}
\floatname{algorithm}{Algorithm}

\mdfsetup{linewidth=0.8pt}

\title{Efficient Computation of Confidence Sets Using Classification on Equidistributed Grids}
\author{Lujie Zhou\footnote{Department of Economics, Pennsylvania State University. lbz5158@psu.edu}}
\date{\normalsize First version: October 2023 \\
This version: November 2024}
\begin{document}
\maketitle

\begin{abstract}
Economic models produce moment inequalities, which can be used to form tests of the true parameters. Confidence sets (CS) of the true parameters are derived by inverting these tests. However, they often lack analytical expressions, necessitating a grid search to obtain the CS numerically by retaining the grid points that pass the test. When the statistic is not asymptotically pivotal, constructing the critical value for each grid point in the parameter space adds to the computational burden. In this paper, we convert the computational issue into a classification problem by using a support vector machine (SVM) classifier. Its decision function provides a faster and more systematic way of dividing the parameter space into two regions: inside vs. outside of the confidence set. We label those points in the CS as 1 and those outside as -1. Researchers can train the SVM classifier on a grid of manageable size and use it to determine whether points on denser grids are in the CS or not. We establish certain conditions for the grid so that there is a tuning that allows us to asymptotically reproduce the test in the CS. This means that in the limit, a point is classified as belonging to the confidence set \textit{if and only if} it is labeled as 1 by the SVM. 

\end{abstract}

\thispagestyle{fancy}

\section{Introduction}
\indent \par
The confidence regions are of central importance in a wide range of applied work spanning many fields of studies. While its computation is not the focus of in the context of economics and econometrics as much as its theoretical construction, to be able to construct the confidences with modern computing tools is a necessary component of the empirical economic analysis. 

We see that oftentimes structural modeling leads to a characterization of the confidence region as a set of parameter values that satisfy a certain inequality or a vector of inequalities (i.e. the criterion inequalities, or simply the test). To perform the pointwise testing when inverting the test to obtain an analytical expression of the confidence region for the parameter of interest is not feasible, one would break the parameter space, typically a subset that lies in some Euclidean space, into very dense grid points and evaluate the criterion to keep the points accepted by the test. This approach is known as the exhaustive searching. For example, if one is interested in the inference of a scalar-valued parameter, the easiest thing to do is to cut (maybe some segment of, based on the prior knowledge) the real line into, say, 10,000 points of equal gaps, followed by plugging those values into the criterion inequality to obtain the confidence interval(s) once we have the criterion inequality defined. Now imagine if we would like to apply the same idea to the computation of parameter vectors of higher dimensions. With two or three parameters, so $10,000^2$ and $10,000^3$ grid points, we would probably do just as well in that we can easily visualize the regions in which parameter vectors all satisfy the criterion inequality, provided that computation efficiency is not really a concern to economists. But we can already see that the computation complexity grows exponentially fast in the number of parameters of interest, and very quickly it will expand out of a manageable scale if one would like to maintain the same precision along each dimension. Especially when the test statistic is not asymptotically pivotal, as we will see in the setup in the next section, an additional computation of the cutoff value has to be done within the evaluation of every grid point.

A second issue that arises is the reporting of these confidence sets. Even though we have formed fine grids and obtained those regions after having the program run for hours or days, the researcher would have hard time describing those regions or present them simply because we cannot visualize objects beyond three dimensions. After all, the grid points along do not tell us about the shape, angle, or connectedness of the confidence regions. 

In this paper, we will explore the number-theoretic notion of equidistributed sequences and integrate it with a modern machine learning method, support vector machine (SVM), originally proposed in \cite{CV-1995}, to address the computation challenge described above. The essence of this approach is, by turning the computation problem into a classification problem, we build a classifier that is fast to evaluate and mimics the decision behavior of the test using a moderately-sized grid, and apply it to a much larger grid to get fast results, hence accelerating the computation of dealing with the gigantic number of grid points while still sustaining the coverage of the underlying the confidence regions. Additionally, this approach is also compatible with parallel computing in that we can parallelize the computation of both grid evaluation and the fitting of the SVM classifier (see, e.g. \cite{svm_parallel}), should this situation become relevant. The main contribution of this paper includes introducing the general notion of equidistributed sequences to the problems of grid search; most importantly, we convert the problem into a classification task and show that it is possible to tune the SVM classifier, which is fast and easy to use once trained on a pre-specified grid, such that it exactly preserves the asymptotic coverage of the confidence sets on the equidistributed grids. Although this paper does not aim to solve the issue of reporting the confidence sets, we hope to venture a different angle around this issue via the use of the decision function associated with the SVM classifier. One can conservatively report the smallest box that contains all grid points predicted by SVM classifier to be in the confidence set, which is characterized by 2 times the dimension of the space numbers, i.e. maximum and minimum along each dimension. Alternatively, we propose reporting the decision function after training the classifier, which is concise, very easy to work with, and conserves all the information in the classifier. 

The rest of this paper proceeds as follows. The first prat of Section \ref{lit_review} introduces the framework where a computational approach becomes necessary when obtaining confidence sets. The second part then provides a discussion of some work that involves moment inequality confidence sets alongside the methods used in the econometrics literature. Section \ref{methods} describe the idea of support vector machine classifier and its asymptotic behavior, followed by the conditions and tuning that ensure the preservation of coverage, which are the main contribution of this paper. Section \ref{simulation} provides simulation results in graphs to help visually illustrate the performance of SVM classifiers. Section \ref{conclusion} concludes. Details of the equidistributed sequences and all proofs can be found respectively in Appendix \ref{grid} and Appendix \ref{proofs}.

\section{Setup and Related Literature}\label{lit_review}
\subsection{Confidence Sets}
The general setup follows the framework provided in \cite{Rosen2008} that describes the construction of the confidence regions from models based on moment inequalities. And as is mentioned in the paper, similar construction is also seen in a number of other papers, such as \cite{Wolak1991}, \cite{Ciliberto_Tamer}, \cite{ALG-2016}, \cite{Li_Henry}, etc. In particular, Rosen defines the set to be all parameter values $\theta\in\Theta$, the parameter space, at which some criterion value is less than a threshold value, $Z_{1-\alpha}$, corresponding to the significance level $\alpha$, and it is given as follows. 
\begin{equation}\label{eq1}
C R_{1-\alpha} \equiv\left\{\theta \in \Theta: n \hat{Q}_{n}(\theta) \leq Z_{1-\alpha}\right\},
\end{equation}
where the parameter-dependent component of the criterion is given by
\begin{equation}
\hat{Q}_{n}(\theta)=\min _{t \geq 0}\left[\hat{E}_{n}[m(y, x, \theta)]-t\right]^{\prime} \hat{V}_{\theta}^{-1}\left[\hat{E}_{n}[m(y, x, \theta)]-t\right],
\end{equation}
and $n\hat{Q}_n(\theta)$ typically follows some chi-bar-square distribution, a mixture of chi-square distributions. Here in this Wald-type statistic, $m(\cdot)$ denotes the finite-dimensional vector of moment conditions which the economic model predicts to be non-negative, and $\hat{V}_\theta$ is the sample variance of $m(y,x,\theta)$ with $x$ and $y$ being the covariate and outcome variables. The idea of this statistics is that if $\theta$ belongs to the identified region of the model, $\Theta^*$, then this quantity should be very small. The cutoff value is computed in a way such that 
\begin{equation}
\inf _{\theta \in \Theta^{*}} \lim _{n \rightarrow \infty} \operatorname{Pr}\left\{n \hat{Q}_{n}(\theta) \leq Z_{1-\alpha}\right\} \geq 1-\alpha
\end{equation}
The computational complication arises for the models where the test statistic is not asymptotically pivotal. Specifically, the asymptotic distribution of the test statistic may depend on those moments with expected value $0$. This is resolved by taking upper bounds of the number of such moment conditions and construct conservative confidence sets for $\theta_0$. While this is a very clear setup in theory, the specific steps of computation given the precision of the confidence region and the parameter set to consider are less obvious. The quote below from \cite{Rosen2008} discusses the computation of the confidence sets, and shows the sense in which the choices can be ad hoc.
\begin{displayquote}
\textit{Appropriate choice of grid values $G$ depends on the particular application. How fine the grid should be depends on the desired level of precision for $C R_{1-\alpha}$. If $\Theta^{*}$ is known to be sufficiently regular (e.g. closed and convex), certain values of $\theta$ may be able to be included or discarded without explicitly evaluating $n \hat{Q}_{n}(\theta)$.}
\end{displayquote}
Our approach on the computation is relatively more straightforward and less case dependent compared to \cite{Rosen2008}, so long as a concrete and meaningful definition of the critical region can be written down both on paper and in computer code.

\subsection{Literature Review}
In a recent work, \cite{KKLS-2019} employ a slight variation of the standard grid search; that is, they limit the searching to certain regions of the real line (the parameter space in their case). This limited grid search scheme begins with, say, $\hat{\theta}^U+0.01$ where $\hat{\theta}^U$ is the estimated upper bound of the identified set for the true parameter $\theta$, proceeds to $\hat{\theta}^U+0.02$ if some null hypothesis is not rejected, and stops when it is rejected. Similar procedures are also applied to the estimated lower bound. This guarantees that all points beyond the stopping values, $\theta^l$ and $\theta^u$, are ruled out, hence the asymptotically uniformly valid $1-\alpha$ confidence set for $\theta$ is given by the interval $[\theta^l,\theta^u]$. They also propose that the stochastic search can be used to find good directions for the parameter of interest in the higher dimensional setting. This algorithm adds perturbation during each iteration of the searching. A risk that the authors recognize is that there is a real risk for the random algorithm to exit the identified set, without exploiting the specific structure of the problems. However, it is not immediately clear how to extend this method to parameters in higher dimensional spaces because, as is mentioned above, putting together the confidence intervals along each individual dimension does not give us the correct coverage probability (e.g. the ``rectangle'' versus the ``ellipse''). On the contrary, our approach handles higher dimensions especially well with the use of SVM because the hyperplane it draws lives in a transformed space and is not constructed in a successive manner.

The family of intelligent algorithms can be useful when solving for global maximum or minimum, such as genetic algorithm, particle swarm optimization algorithm, artificial fish swarm algorithm, artificial bee colony algorithm, and firefly algorithm. The wolfpack algorithm introduced in \cite{WZ-2014}, for example, first obtain the minimizer of as in equation (\ref{eq1}), and then perform a grid search near the global minimizer. The obvious advantage is the massive parallelizability in that we can multi-process the same algorithm to visit various regions of the space via different initializations. Nonetheless, there is no guarantee that all disconnected regions of the confidence set will be visited since one starts grid searching only around the global extremum after discarding other information along the search of the extremum, and one certainly does not get around the issue of performing a grid search even after nailing the global minimum. Working equally well with this idea is another family of algorithms, the derivative-free optimization (DFO) algorithms. In comparison, our new approach does not hinge on the fact that a global extremum is found first as a guiding direction; more importantly, through the use of the equidistributed sequences, we have control over the way we visit different regions of the parameter space. We can just as easily parallelize the computation, too. The algorithm below describes the detailed steps of the intelligence algorithm, where ``wolves'' refer to the nodes that run in a parallel manner, and the ``lead wolf'' is defined to be the node with the best objective value in the current iteration. 

\begin{algorithm}
\begin{algorithmic}[1]
\STATE Initialize the algorithm parameters: positions/number of wolves, maximum iteration, etc.
\STATE All but the lead wolf search until a higher objective function value is attained (then replace the lead wolf) or maximum iteration is reached.
\STATE All wolves move toward the lead wolf. If a new lead wolf emerges, go to step $2$; otherwise keep moving until they are ``close'' to the lead wolf. 
\STATE All but the lead wolf take besieging behavior in the vicinity of the lead wolf.
\STATE Update the leader when appropriate. And update the wolf population to open up new search space.
\STATE Stop if one stopping condition is met; otherwise, go to step $2$.
\end{algorithmic}
\caption{\cite{WZ-2014}: Wolf Pack Algorithm}
\label{alg}
\end{algorithm}

Moment inequalities arise in many different scenarios, such as relaxing model assumptions, missing data, models with multiple predictions or solutions, etc. \cite{manski_tamer} represent a particular missing data issue in terms of moment inequalities. Specifically, when upper and lower bounds of outcome data are observed rather than the actual outcome, a set of conditional moment inequalities can be derived that must hold with probability 1. Large amount of papers have focused on the models with partially identified parameters. In the case where the univariate parameter of interest is not point-identified in the model, \cite{imbens_manski} provides a way of constructing confidence intervals, for the particular parameter rather than the entire identified region, that has the correct pre-specified coverage probability. \cite{CHT2007} and \cite{Romano_Shaikh} were the earliest attempts who separately develop inference procedures for both the multi-dimensional set of identified parameter vectors and the true parameter vector of interest, both of which rely on subsampling methods that may be computationally difficult. The recent work by \cite{Li_Henry} proposes a framework with incomplete models in which the confidence set obtained by inverting the test displays exact finite sample coverage. Notably, the test statistic and critical value are computed through a discrete optimal transport problem, and they also propose a fast preliminary search methodology that locates the region in which the true parameter vector lies, thus alleviating the cumbersome search over the entire parameter space. 

There is a literature on sub-vector inference for moment inequalities when the parameters of interest form only a sub-vector aside from the nuisance parameters. In particular, if the conditional moment inequalities satisfies linearity conditions, one can exploit this structure and efficiently compute the profiled maximum statistic as in \cite{ARP2019} or the profiled quasi-likelihood ratio statistic proposed in \cite{cox_shi_2022}, regardless of the dimension of the nuisance parameters.

\cite{ALG-2016} develop a structural model of chain-story entry with discrete intensity and obtain the asymptotic properties of the pivotal test statistic via kernel-based estimation, which is then used to construct the confidence set. A grid search is then applied in the parameter space to determine numerically the region of acceptance via the criterion defined by the authors. \cite{Ciliberto_Tamer} study the effect of general firm heterogeneity on the market structure in the context of U.S. airline industry, and propose a moment inequalities-based procedure to conduct inference on the parameters of interest. They employ simulated annealing algorithm to approximate the global minimum of the objective criterion along with subsampling to obtain the critical value (as described in \cite{CHT2007}) in the construction of the confidence region. We discuss more details of their approach in Subsection \ref{simulation_cht}.

Finally, equidistributed sequences are no strangers to the studies of computation in economics. The computational methods textbook, \cite{Judd1998}, introduced this concept in the context of numerical integration. It is noted that equidistributed sequences display ex post uniformity, which is a property that yields asymptotically valid approximation to any integral. \cite{Judd1998} also discusses that the precision is lost as more points along the same sequences are generated because they are sensitive to the round-off errors due to the machine precision. Additionally, \cite{miranda_fackler} builds on the aforementioned formulation for Monte Carlo integration and offers a function that computes the equidistributed sequences along with their weights in the well-known \textit{CompEcon} toolbox. They point out that certain such sequences achieve uniformity in both an ex ante and an ex post senses, hence better representing the regions under investigation than truly random sequences.

\section{Methodologies}\label{methods}
In this section, we will first describe the mechanism of the support vector machine classifiers. Then, we introduce the asymptotic property of the SVM classifiers when the tuning parameters take certain values. And lastly, we provide conditions based on the equidistributed sequences which guarantee that our procedure replicates the behavior of the limiting confidence sets. Our complete approach would be to first generate a moderately-sized set of grid points (or more, depending on the researcher's time budget) using one of the equidistributed sequences, evaluate these points using the criterion inequalities, assign binary labels to these grid points (i.e. 1 if in the confidence region; -1 otherwise), and fit a SVM classifier to these labeled grid points. The tuning of the SVM parameters will always be such that the SVM classifier predicts all training examples correctly. As a result, we end up with a trained SVM classifier on the grid points generated with the equidistributed sequences, with which we can then assign binary labels to more grid points that are outside of the grid we use to train the classifier. This approach is much more efficient compared to exhaustively evaluating the grids, because training the SVM is a one-time investment and its computation is relatively negligible compared to evaluating the grid points in most applications. So in the extreme case where we evaluate just as many grid points as exhaustive search, SVM gives us a systematic way of replicating the confidence set and extrapolating the prediction to larger grids via its decision function. Most importantly, it allows us to re-produce the confidence sets with precision in the limit if the grids meet the very mild conditions described below. 

Other classification methods could work too in this setting, but there are other drawbacks. For example, convolutional neural networks are also excellent at approximating functions but the theoretical result puts restrictions on the type of functions it can perfectly approximate, such as continuity on a box in $\mathbb{R}^d$ (see \cite{universal_approx}). Another example is the nearest-neighbor classifiers, which inefficiently computes the distance between a new point and all training points although it requires no fitting. Note that logistic regression does not apply because this is a non-linear problem in general.

\subsection{Support Vector Machine Classifiers}\label{svm}
Support vector machines, which is of core importance to this paper, are ideal methods for classification tasks. In a simple case of points with binary labels (hence defining two groups) along the real line, SVM finds the boundary point such that it best separates the two groups; mathematically, it finds the point that minimizes the incorrect labels on both sides of the point. When there are disconnected groups, this problem becomes non-linear. The kernel method is applied then to transform the problem into a linear problem in other spaces. Essentially, we transform the data via a kernel function with tuning parameter such that the data becomes linearly separable in the new space. To exemplify its working, let's consider a 2-dimensional case where one group is surrounded by another group, as in the left panel of Figure \ref{fig1} below. Clearly, one cannot draw a single straight line to separate the two groups. Kernel method adds a third dimension, altitude, to the data, and as we see in the right panel of Figure \ref{fig1}, it is straightforward to draw a plane (i.e. a line equivalent in 3-dimensional space) that completely separates the two groups. In fact, there are a few possible planes one can draw to accomplish the same result. SVM finds the one that maximizes the minimum distance to both groups. Now if we project this plane back to the original data space, we would end up with a closed curve resembling that one would draw to separate the groups by hand. This is called the \textit{decision boundary}. We will explore more of this in Section \ref{simulation}

\begin{figure}[H]
    \centering
    \includegraphics[scale=0.55]{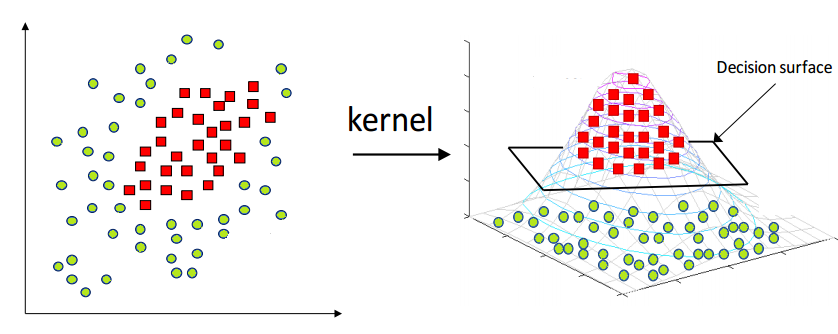}
    \caption{The mechanism of SVM classifiers in a simple 2D example. (\href{https://medium.com/analytics-vidhya/how-to-classify-non-linear-data-to-linear-data-bb2df1a6b781}{\textit{Source}})}
    \label{fig1}
\end{figure}

Since in reality, the kernel method only runs the computations \textit{as if} the data are in the higher dimensional space without actually transforming them, this improves the efficiency and makes non-linear classification possible computationally. Aside from the fast computation, SVM classifiers also have great performances separating the groups very consistently. And for the same reason, it is very suitable when our object lives in highers dimensions. Again, as is discussed earlier, the appealing visual representation is limited to spaces of dimensions no more than three. But it does allow us to have an alternative nonparametric and computational characterization of the very same confidence regions defined by our economic models. Future research along this line is appropriate to exploit the value of the decision boundaries created by the SVM classifiers.

To fully understand the way SVM works and discuss the ways in which we can adapt it to our applications, we will look into the mathematics in the background next. Once we equip the mechanism with computation tools, we should be able to relax the standard approach and modify it with our own criteria. The following parts build the support vector network progressively in three steps, after which a brief discussion of the computational implementation is provided; and at last, we will see where the adaptation is appropriate to extend it in the economic models. For our purposes, it suffices to discuss the cases with two classes labeled $1$ and $-1$, hence a binary classification problem.

\vspace{0.3cm}
\noindent\textit{Optimal Hyperplane} 

Let us first consider the case where the data are linearly separable, i.e. can be separated perfectly with a straight line or its equivalent in higher dimensions. Denote the features $\mathbf{x}_i\in\mathbb{R}^d$ and the labels $y_i\in\{-1,1\}$, for $i=1,\cdots,S$. Formally, we say the classes are \textit{linearly separable} if there exist a vector $\mathbf{w}\in\mathbb{R}^d$ and a scalar $b$ such that 
\begin{equation}\label{eq4}
\begin{aligned}
    &\mathbf{w}\cdot \mathbf{x}_i+b\geq1,\text{ if }y_i=1\\
    &\mathbf{w}\cdot \mathbf{x}_i+b\leq-1,\text{ if }y_i=-1.
\end{aligned}
\end{equation}
Define the \textit{margin} to be the smallest distance from one class to the other, and the \textit{decision boundary} is used to separate the classes. In Figure \ref{fig:margin} below, an example with $2$-dimensional features is presented. One class is labelled as circles and the other is crosses. The margin is given by the perpendicular distance from the highlighted circle to the highlighted crosses. And the decision boundary is the dotted line in between the two solid lines.
\begin{figure}[H]
    \centering
    \includegraphics[scale=0.5]{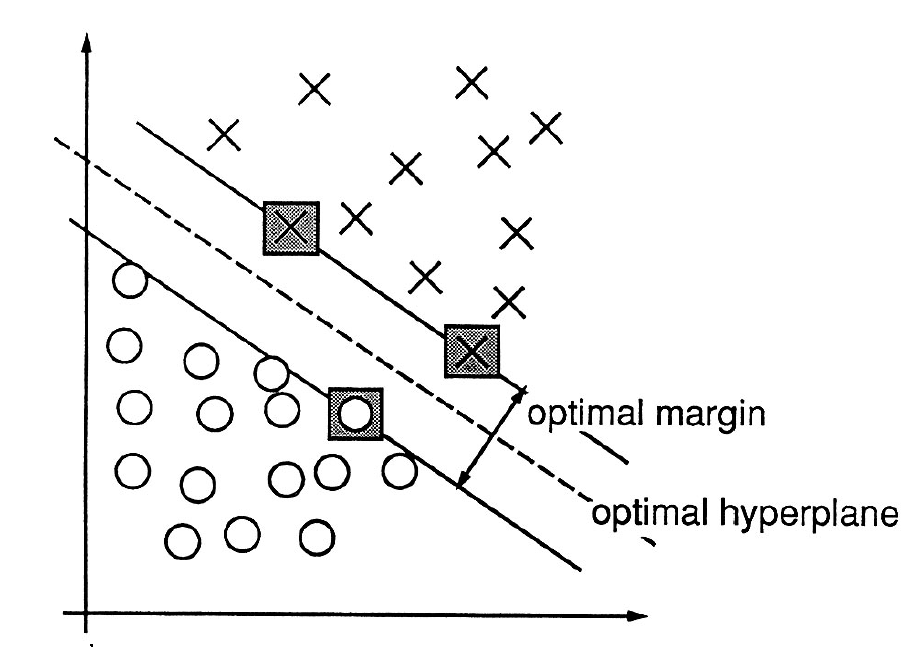}
    \caption{Decision boundary and margin.}
    \label{fig:margin}
\end{figure}
Having this in mind, we proceed to define the following concepts. Since there are infinitely many lines one can draw, in the above example, to separate the two classes, we take the line that yields the maximal margin to be our \textit{optimal hyperplane}, and call the largest margin the \textit{optimal margin}. In other words, optimal hyperplane maximizes its distance to the nearest data point of either class. Mathematically, we write equation (\ref{eq4}) as, for $i=1,\cdots,S$,
\begin{equation}\label{eq5}
    y_i(\mathbf{w}\cdot \mathbf{x}_i+b)\geq1,
\end{equation}
and the optimal hyperplane
\begin{equation}
    \mathbf{w}_0\cdot \mathbf{x}+b_0=0,
\end{equation}
which uniquely separates the two classes and maximizes the margin, i.e.
\begin{equation}
    (\mathbf{w}_0,b_0)=\arg\max\left\{ \min_{\{\mathbf{x}:y=1\}}\frac{\mathbf{x}\cdot\mathbf{w}}{||\mathbf{w}||}-\max_{\{\mathbf{x}:y=-1\}}\frac{\mathbf{x}\cdot\mathbf{w}}{||\mathbf{w}||} \right\}
\end{equation}
By equation (\ref{eq5}), the maximal margin is given by $\frac{2}{||\mathbf{w}_0||}=\frac{2}{\sqrt{\mathbf{w}_0\cdot\mathbf{w}_0}}$. So, constructing the optimal hyperplane is a quadratic programming problem that minimizes $\mathbf{w}\cdot\mathbf{w}$ under constraint (\ref{eq5}). Equivalently, we have the following convex minimization problem
$$\min\,\frac{1}{2}||\mathbf{w}||^2\quad\text{subject to}\quad y_i(\mathbf{w}\cdot\mathbf{x}_i+b)\geq1, \forall i.$$
The Lagrangian is given by 
\begin{equation}\label{eq8}
    L(\mathbf{w},b,\bm{\alpha})=\frac{1}{2}||\mathbf{w}||^2-\sum_{i=1}^n \alpha_i(y_i(\mathbf{w}\cdot\mathbf{x}_i+b)-1)
\end{equation}
Importantly, the Lagrange multiplier $\alpha_i=0$ for all data point $i$ that satisfies the constraint (\ref{eq5}) with strict inequality, by Karush-Kuhn-Tucker Theorem. Furthermore, the KKT first-order conditions are as follows
\begin{equation}\label{eq9}
\begin{aligned}
\mathbf{w}&=\sum_i\alpha_i y_i\mathbf{x}_i\\
0&=\sum_i\alpha_i y_i
\end{aligned}.
\end{equation}
Put together the above, we have that the optimal hyperplane is only determined by the points that satisfy constraint (\ref{eq5}) with equality, i.e. the \textit{support vectors}. Plugging equations (\ref{eq9}) into the Lagrangian (\ref{eq8}), we have the Lagrangian dual below that only operates over values of $\alpha_i$ under constraint (\ref{eq5}).
\begin{equation}\label{eq10}
    L(\bm{\alpha})=\sum_i\alpha_i-\frac{1}{2}\sum_i\sum_j\alpha_i\alpha_j y_i y_j (\mathbf{x}_i\cdot\mathbf{x}_j)
\end{equation}
Moreover, denote the maximizer $\bm{\alpha}^0$, the separating hyperplane is given by
\begin{equation}\label{eq13}
    0=\sum_i\alpha_i^0 y_i (\mathbf{x}_i\cdot\mathbf{x}) + b_0
\end{equation}
Note that $\mathbf{w}_0$ can be solved from top equation in (\ref{eq9}), and $b_0$ is typically solved using the average of all points such that (\ref{eq5}) holds with equality for numerical stability (\cite{Elements}). Once the solution is obtained, we take the sign of the function $f(\mathbf{x}^\dagger)=\mathbf{w}_0\cdot\mathbf{x}^\dagger+b_0$ as the classification of an new vector $\mathbf{x}^\dagger$ with unknown label. We will see next that the dot product terms in equations (\ref{eq10}) and (\ref{eq13}) are replaced with some kernel function to extend this idea to nonlinear classes.

\vspace{0.3cm}
\noindent\textit{Soft Margin Classifier}

To make our technique even more flexible and handle cases where classes overlap, we will allow for some tolerance of error. Define $\xi_i\geq0$ to be the tolerance associated with classifying each observation incorrectly. Our objective becomes to minimize the errors while finding the optimal hyperplane. This implies that we can combine the old problem with the new component into the following formulation.
\begin{align}
&\min \frac{1}{2}\mathbf{w}^2+C\cdot F\left(\sum_i \xi_i\right) \nonumber\\
\text{subject to}\quad & y_i(\mathbf{w}\cdot\mathbf{x}_i+b)\geq 1-\xi_i,\quad i=1,\cdots,S \\
&\xi_i\geq0,\quad i=1,\cdots,S
\end{align}
Here, the coefficient $C$ is a tuning parameter which determines the trade-off between maximizing the margin and making sure that all features $\mathbf{x}_i$ lie on the correct sides of the margin. Function $F$ is a monotone convex function that takes form of some loss function, such as \textit{hinge loss} given by $\max\{0,1-y_i(\mathbf{w}\cdot\mathbf{x}_i+b)\}$. If we think of the first component of this objective as the $L_2$ regularization term (or penalty term), then SVM can also be viewed as penalized regression under the hinge loss function. The above problem is known as the \textit{primal problem}. The classical approach to computing the soft margin classifier is to solve the \textit{Lagrangian dual} of the above problem and turn it into a quadratic programming problem. There are more recent approaches such as sub-gradient descent and coordinate descent developed for this problem, whose details I have not gotten into yet. 
Note it can be showed that the solution $\mathbf{w}_0$, same as the optimal hyperplane problem, can be written as a linear combination of the support vectors, i.e. 
\begin{equation}
    \mathbf{w}_0=\sum_i \alpha_i^0 y_i\mathbf{x}_i
\end{equation}
where $\bm{\alpha}^0$ is solved from the dual problem.

\vspace{0.3cm}
\noindent\textit{Support Vector Machine}

Suppose we have a non-linear problem. For some given transformation $\phi:\mathbb{R}^d\to\mathbb{R}^p$, which we know/choose a priori, that maps from the original input space to some feature space, we could separate the classes linearly in the feature space very easily. As discussed above, a coherent way to extend the above method is to replace the $\mathbf{x}_i$ with $\phi(\mathbf{x}_i)$, thus the dot products $\mathbf{x}_i\cdot\mathbf{x}_j$ replaced by $K(\mathbf{x}_i,\mathbf{x}_j)\equiv\langle\phi(\mathbf{x}_i),\phi(\mathbf{x}_j)\rangle$. We call this generalized inner product, $K$, \textit{kernel function}. More specifically, we would like to learn some normal vector $\mathbf{w}\in\mathbb{R}^p$ and scalar $b$, and according to the above formulation, we have 
\begin{equation}
    \mathbf{w}=\sum_i \alpha_i y_i\phi(\mathbf{x}_i),
\end{equation}
which then implies
\begin{equation}
\begin{aligned}
    f(\mathbf{x})&=\mathbf{w}\cdot\phi(\mathbf{x})+b \\
    &=\sum_i \alpha_i y_i \phi(\mathbf{x}_i)\cdot\phi(\mathbf{x})+b \\
    &=\sum_i \alpha_i y_i K(\mathbf{x}_i,\mathbf{x})+b.
\end{aligned}
\end{equation}
While the computation takes place in the transformed feature space now, the steps to compute the classifier (i.e. the support vectors and the weights $\alpha_i$) follow exactly from before. We require this kernel $K$ to be a symmetric, positive semi-definite function so that Mercer's theorem applies. This is because only kernels that satisfy Mercer's condition can be decomposed into inner products of some underlying map $\phi$ (i.e. such $\phi$ exists). This is known as the ``kernel method'', and it allows us to work with potentially infinite dimensional transformation of the input data without having to explicitly transform them. Some common choices of the kernel function include the polynomial kernel of degree $d$\footnote{Derivation: \href{https://en.wikipedia.org/wiki/Polynomial_kernel}{here}}, which is given by 
\begin{equation}
    K(\mathbf{u},\mathbf{v})=(\mathbf{u}\cdot\mathbf{v}+1)^d;
\end{equation}
and radial basis function kernel (Gaussian)\footnote{Derivation: \href{https://en.wikipedia.org/wiki/Radial_basis_function_kernel}{here}}
\begin{equation}
    K(\mathbf{u},\mathbf{v})=\exp\left(-\frac{||\mathbf{u}-\mathbf{v}||^2}{2\sigma^2}\right)
\end{equation}
where $\sigma>0$ is a free parameter.

The implementation of the SVM has been well developed and packaged in a number of machine learning toolkits, including LIBSVM, Matlab, scikit-learn in Python, etc. For the simulation exercises in section \ref{simulation} below, I used the ``e1071''\footnote{See \href{https://cran.r-project.org/web/packages/e1071/e1071.pdf}{https://cran.r-project.org/web/packages/e1071/e1071.pdf}.} library in R for convenience. When implementing our own variant of the SVM, we can make use of the source code of some of the these packages.

\subsection{Asymptotic Behavior}
Because of some properties of the radial basis function (RBF) kernel, we should pair our method with this choice of kernel. The most desirable property is that we can achieve 100\% perfect classification accuracy on the training data with the RBF kernel, provided that no two points of different classes lie on top of each other. \cite{KL-2003} show that for a fixed parameter value $C$ that is sufficiently large, so long as we make the $\sigma$ parameter really small (e.g. much smaller than the smallest distance between any two data points), we have that for any $i$, the predicted label is $\hat{y}(x_i)=y_i$. The classifier classifies the training example correctly in an asymptotic sense of $\sigma$,  given a fixed $C$. Formally, we have the following result. 

\begin{assumption}\label{ass1}
\begin{enumerate}
    \item[]
    \item $l_0>l_1+1>2$ where $l_0$ and $l_1$ are the number of training examples in class 0 and class 1, respectively.
    \item For $i\neq j, x_i\neq x_j$. That is, no two examples have identical $x$ vectors.
\end{enumerate}
\end{assumption}

\begin{lemma}\label{perfect_fit}
Under Assumption \ref{ass1}, if $C>l_0/(l_0+l_1)$, then, as $\sigma^2\to0$, the classifier function 
$$f(\mathbf{x})=\mathbf{w}_0\cdot\phi(\mathbf{x})+b_0$$ 
where $\phi(\cdot)$ is implicitly defined through the RBF kernel, i.e. $\phi(\mathbf{x})\cdot\phi(\mathbf{x}')=\exp\left(-\frac{||\mathbf{x}-\mathbf{x}'||^2}{2\sigma^2}\right)$, classifies all training examples correctly and classifies the rest of the space as class 0.
\end{lemma}

As a consequence of the Lemma, we will always pair our SVM classifier with the RBF kernel. Throughout this paper we will assume Assumption \ref{ass1} holds and that the tuning parameter $C$ satisfies the condition given in Lemma \ref{perfect_fit}, hence the tuning will be focused on the parameter $\sigma^2$.  Albeit, the above result creates an issue of overfitting, which extensively undermines the out-of-sample predictive power. However, as far as our problem is concerned, this becomes an ideal attribute that allows us to capture the confidence regions. To illustrate the idea of this result, let us consider the following special case of the classification function 
\begin{equation}
f(\mathbf{x})=\sum_i \alpha_i y_i K(\mathbf{x}_i,\mathbf{x}),
\end{equation}
where the kernel is given by 
$$K(\mathbf{x}_i,\mathbf{x})=\exp\left(-\frac{||\mathbf{x}_i-\mathbf{x}||^2}{2\sigma^2}\right).$$ 
When we make $\sigma$ really small, it follows that $K(\mathbf{x}_i,\mathbf{x})\approx0$; except when $\mathbf{x}=\mathbf{x}_i$, $K(\mathbf{x}_i,\mathbf{x}_i)=1$. Therefore, the predicted label of the $j$th training point is 
$$\hat{y}(\mathbf{x}_j) = \sgn\left(\alpha_j^0 y_j K(\mathbf{x}_j,\mathbf{x}_j)+\sum_{i\neq j}\alpha_i^0 y_i K(\mathbf{x}_i,\mathbf{x}_j)\right)\approx \sgn(\alpha_j^0 y_j)$$
Now if we simplify $\alpha_i^0$ to $1$s, then clearly $\hat{y}(\mathbf{x}_j)$ would have the same sign as $y_j$. The immediate implication for our purposes is that as we increase the number of training points $n$ in our setup, since we can always completely classify the two groups (in and outside the confidence set) by taking a sequence of scalars $\{\sigma_n\}_{n\in\mathbb{N}}$ with $\sigma_n\to0$ as $n\to\infty$, we can perfectly recover the true underlying confidence set. We will formally state this result in the next subsection. Figure \ref{fig:overfit} shows the issue of overfitting when the $\sigma$ parameter value is too small, which puts too much weight on correctly classifying the training example. Whereas this can be mitigated usually via cross validation in practice, we will take full advantage of this behavior to obtain the confidence sets of interest. 
\begin{figure}
    \centering
    \includegraphics[scale=0.55]{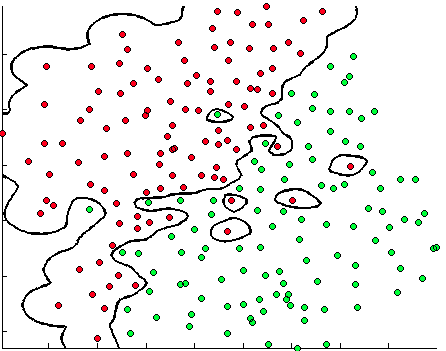}
    \caption{Overfitting issue with the Gaussian RBF kernel. (\href{http://openclassroom.stanford.edu/MainFolder/DocumentPage.php?course=MachineLearning&doc=exercises/ex8/ex8.html}{\textit{Source}})}
    \label{fig:overfit}
\end{figure}

Furthermore, RBF can be interpreted as if it maps a vector in the the original data space to an infinite-dimensional feature space where a separating hyperplane that perfectly divides the two groups is guaranteed to exist, although we never need to explicitly transform the data into the infinite-dimensional features due to the kernel method. To see this, let $\mathbf{u}$ and $\mathbf{v}$ be any two vectors, setting $2\sigma^2=1$, we have
\begin{align*}
K(\mathbf{u},\mathbf{v}) &= \exp(-||\mathbf{u}-\mathbf{v}||^2)\\
&=\exp(-\langle \mathbf{u}-\mathbf{v},\mathbf{u}-\mathbf{v}\rangle)\\
&=\exp(-||\mathbf{u}||^2-||\mathbf{v}||^2+2\mathbf{u}\cdot\mathbf{v})\\
&=\exp(-||\mathbf{u}||^2-||\mathbf{v}||^2)\sum_{n=0}^\infty \frac{(2\mathbf{u}\cdot\mathbf{v})^n}{n!},
\end{align*}
by Maclaurin series. Therefore, denoting $M=\exp(-||\mathbf{u}||^2-||\mathbf{v}||^2)$, we can define the transformation function from
$$\phi(\mathbf{u})\cdot\phi(\mathbf{v})=M\,\sum_{n=0}^\infty\frac{(2\mathbf{u}\cdot\mathbf{v})^n}{n!}.$$
In the case where $u$ and $v$ are scalars, we can write it out as follows 
$$\phi(u)=\sqrt{M}\cdot \bigg\langle 1, \,\sqrt{\frac{2}{1!}}\, u,\,\sqrt{\frac{2^2}{2!}}\,u^2,\,\sqrt{\frac{2^3}{3!}}\,u^3,\,\cdots \bigg\rangle.$$

\subsection{Coverage}
With the above theoretical guarantee, we now show that the desired coverage can be achieved by the SVM via proper tuning of the tuning parameters, along with proper design of the grid in the absence of further knowledge about the true paremeter. In this subsection, we provide the conditions that the grid and the values of the tuning parameters need to satisfy in order to preserve the asymptotic coverage of the underlying confidence sets. Moreover, when we tune the SVM classifier to behave exactly like the confidence set as we will show below, it not only maintains the asymptotic size, but also inherits the power in that any point outside the confidence set will be be labelled $-1$ by the post-tuning SVM classifier. We will look at different cases when sample size $n$ goes to infinity and when it is some finite number. As the sample size grows, we would like to generate increasingly more grid points as well over which we evaluate the test, and feed all grid points to train the SVM classifier. However, this must happen in a way that sufficiently explores the parameter space. The reason is that if the grid points are clustered in some regions, then we cannot precisely learn the boundary of the confidence sets in the limit even as the total number of grid points tends to infinity. More formally, we state in the following definitions the ways in which the grid points should be formed in the absence of any knowledge of the parameter of interest. 

\begin{definition}
    $\Gamma\subset \mathbb{R}^d$ is a $d$-dimensional \textit{box} if it is given by $\Gamma=\prod_{i=1}^d [a_i,b_i]$, where $a_i<b_i$, for all $i=1, \cdots, d$. 
\end{definition}

\begin{definition}\label{def:explore}
A grid $S$ \textit{fully explores} the space of interest $\Theta$ if it is the first $|S|$ terms of an equidistributed sequence on $\Gamma(\Theta)$, the smallest box in $\mathbb{R}^d$ that contains $\Theta$.
\end{definition}
\noindent
We borrow this number-theoretic notion of equidistribution to generalize the formation of grids and ensure that any point in the space can be approached arbitrarily closely by our grid points, an essential component in the theoretical result below. A uniform grid is formed from evenly-spaced equidistributed sequences. Another example of the equidistributed sequences is the Monte Carlo sequences, which is nothing other than putting together independent random sequences along each dimension of the parameter space. We leave the full details and further discussion of the equidistributed sequences along with more examples in Appendix \ref{grid}. We acknowledge that if prior knowledge of the true parameter is available, it would be more efficient to avoid visiting all regions of the parameter space and concentrate the grid points in regions that align with the prior information.

We will focus on the classifier with corresponding decision function given by 
\begin{equation}\label{decision_func}
f(\theta)=\sgn\left(\sum_i \alpha_i^0 y_i K(s_i,\theta)+b_0\right),
\end{equation}
setting the non-zero $\alpha_i^0$'s to 1 and $b_0=0$ for simplicity, so only the effect of the kernel is present. This can be extended to include the optimal values $\alpha_i^0$ and $b_0$ without qualitatively changing the following results. Below is another useful definition in finding the sufficient condition of preserving the asymptotic coverage of the confidence sets.

\begin{definition}
The influence of a grid point $s$ on another point 
$\theta$ \textit{dominates} the influence of all other grid points if $k(s,\theta)>\sum_{s_j\neq s}k(s_j,\theta)$, where $k(\cdot,\theta)$ measures the similarity between a point and $\theta$.
\end{definition}

\subsubsection{Large Sample}
First, we provide the conditions as the sample size goes to infinity. Given the decision function in (\ref{decision_func}), we see that any point $\theta\in CS_n$ will be predicted to have label $1$ by the SVM classifier if the nearest grid point has label $1$ and its influence dominates the influence of all grid points with label $-1$, using Gaussian RBF kernel as the measure of similarity. Therefore, we require the following to be true for the grid.

\begin{assumption}\label{ass:grid}
Grid $S_n$ fully explores the parameter space $\Theta$, for each $n$, based on the same equidistributed sequence. That is, $S_n\subset S_{n+1}$.
\end{assumption}
\noindent
As an immediate result of Assumption \ref{ass:grid}, we have the following lemma which states that eventually the closest grid point to any $\theta$ in the confidence set will be a grid point with label $1$; whereas the distance from any grid point outside the confidence set is bounded away from 0, for any sample size. Given a particular $\theta$, we denote the closest grid point with label $1$, i.e. interior grid point, as $I_n(\theta)$, and the closest grid point with label $-1$ as $E_n(\theta)$.

\begin{lemma}\label{lemma2}
Under Assumption \ref{ass:grid}, for any $\theta\in CS\equiv \underset{n\to\infty}{\lim}CS_n$, $||I_n(\theta)-\theta||\to0$ and $||E_n(\theta)-\theta||\to M>0$ as $n\to\infty$, where $M$ is some fixed number and $||\cdot||$ is the Euclidean norm.
\end{lemma}

Therefore, we have the first main result of this paper as follows.

\begin{theorem} \label{thm1}
Suppose Assumption \ref{ass:grid} holds. Then, for any $\alpha\in(0,1)$, there exists tuning such that $\mathbb{P}(\theta_0\in SVM_+)\geq 1-\alpha$, where $\theta_0\in SVM_+$ means SVM classifier predicts $\theta_0$ to have label $1$.
\end{theorem}
\noindent
This tells us that the SVM classifier can be tuned to preserve the asymptotic coverage of the confidence set of interest, so we know using SVM classifier in place of the test with which we evaluate the parameter values would be correct in the limit. Furthermore, as Lemma \ref{perfect_fit} dictates that SVM classifiers can achieve perfect classification on training examples, we have that it is possible tune the SVM classifier in a way that it completely re-produces the classification behavior of the test from which we construct the confidence set. Formally, we state this result in the following corollary.

\begin{corollary}\label{corollary1}
Under the above assumptions, given any confidence set $CS\subset\Theta$, there exists tuning of the SVM classifier such that a point $\theta\in CS$ if and only if $\theta\in SVM_+$.
\end{corollary}

This finding is very useful in that once we finish training the SVM classifier on some moderately sized grid, we can rely on it to determine the confidence regions on much denser grids and produce about the same results as the original confidence sets. Evaluating points with an SVM is a simple computation that is much more desirable than with the original test involving moment inequalities, which can become very burdensome when the critical values depend on the parameter values.

\subsubsection{Fixed Sample Size}
Let us now fix the sample size to some finite $n$. An easy thing to do is the analogous construction to the above result, which will yield the same re-producing behavior of the SVM classifier. But instead of the correct size $\alpha$, the coverage is only lower bounded by the coverage of the corresponding $CS_n$, i.e. $\underset{\theta\in\Theta_0}{\inf}\mathbb{P}(\theta\in CS_n|\theta)$ ($\Theta_0$ is the collection of parameters values consistent with the economic model), which might not be at least $1-\alpha$ for any pre-specified $\alpha$. And likewise, we require that the grid fully explores $\Theta$ as it grows in cardinality. This is formalized in the theorem below.

\begin{theorem}
Given some finite sample size $n$, if we let the grid $S$ grow, i.e. $|S|\to\infty$, in a way that fully explores $\Theta$ based on an equidistributed sequence (as in Definition \ref{def:explore}), then there exists tuning such that $\{\theta\in CS_n\}$ implies $\{\theta\in SVM_+\}$. \footnote{Proof is a special case of Theorem 1 and Corollary 1, hence is omitted in this version.}
\end{theorem}

Regardless of the preservation of coverage, this serves better as a theoretical check than being a useful result in practice because it defeats the purpose of coming up with efficient computational method if we actually evaluate the test on an infinite grid. Therefore, we would preferably like to have some procedure that is more computationally manageable without inflating the size of the uninformative grid to infinity, while trying to maintain the coverage probability of $CS_n$. We can continue to rely on the leading terms of some equidistributed sequence as the grid since this procedure gives the correct asymptotic coverage as discussed above. Alternatively, we can iteratively add grid points in regions with grid points of both classes to pin down the boundary for a better recovery of the finite-sample confidence set. Once we initialize the grid with an equidistributed sequence, which should have sufficient points along the boundary of $\Theta$, we can investigate the neighborhoods of each grid point with a pre-specified radius (such that there is at least one neighboring grid point in the neighborhood for every grid point). If a neighbor has a label different from the centroid, we add a new grid point between the centroid and the neighbor and obtain the label of this new point. In the next iteration, we repeat the same step for the expanded grid. This continues until some stopping criterion is met, and then we train the SVM classifier on the resulting grid. While this way we have to introduce more parameters such as the initial grid, the neighborhood radius, maximum number of iterations, etc., it allows for more flexibility and leaves case-dependent choices to the researcher, hence could achieve better finite-sample results. It is worth pointing out that, as a drawback, regions of the confidence set $CS_n$ that are undiscovered initially are unlikely to be visited later on, hence how the coverage of the SVM classifier following this procedure compares to that of $CS_n$ is ambiguous. Having and incorporating information about the true parameter $\theta_0$ would help reduce the occurrence of such event. Without knowing better, our hope is that the initial grid would pick up all disconnected regions of the underlying confidence set.

\section{Simulation}\label{simulation}
In this section, we demonstrate the use as well as performance of the SVM classifiers under the Gaussian RBF kernel in two simulation studies where we fit SVM classifiers to (i) a consistent estimator for some pre-specified set and (ii) coefficient estimates of an ordinary least squares (OLS) regression model, using arguably coarse grids especially in the latter case. We will showcase that the effectiveness of our proposed method relies not only on the density of the grid, but more importantly the properties of the inference procedure being reproduced, specifically its convergence rate as a function of sample size. Furthermore, we will validate the idea that, even though the one-time training of the SVM classifier may take some time (which is still trivial compared to evaluating the original criterion), evaluating new grid points using the trained classifier takes negligible amount of time in spaces of at least a few dimensions.

\subsection{Consistent Estimate for a Set}

\begin{figure}[b!]
    \centering
    \includegraphics[scale=0.25]{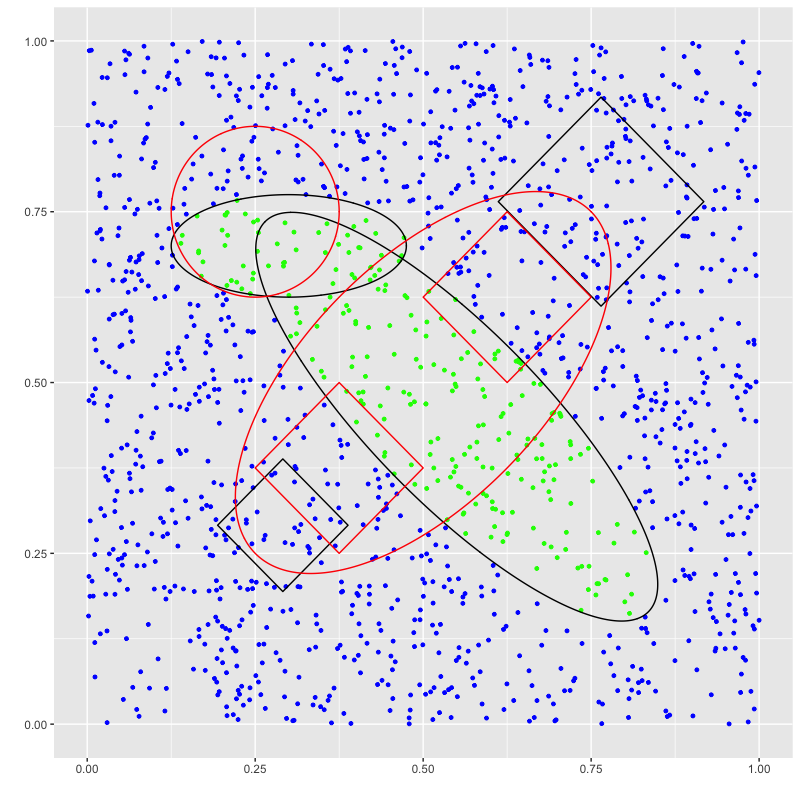}
    \includegraphics[scale=0.25]{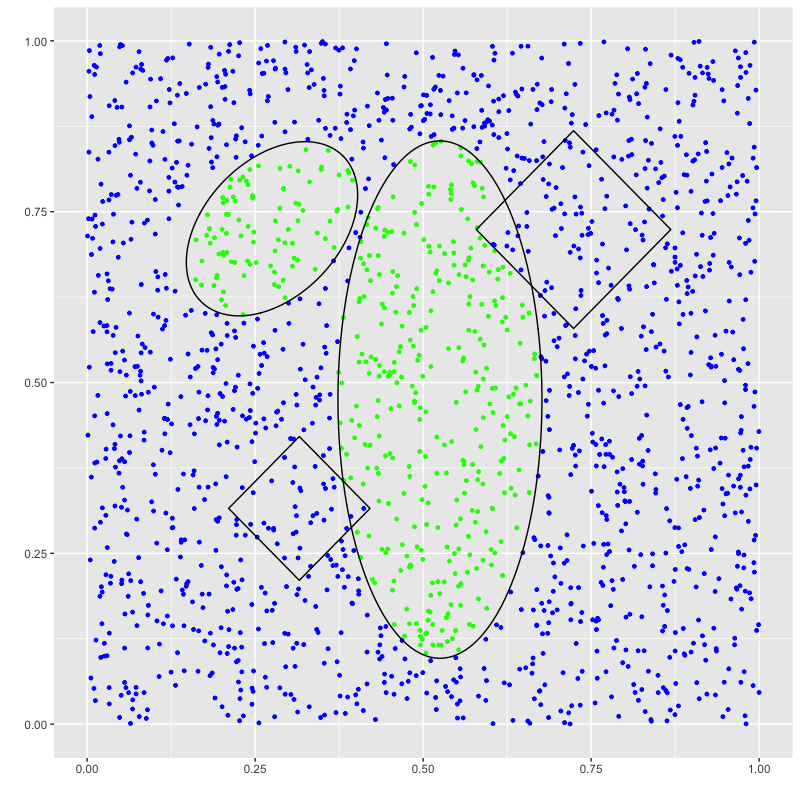}
    \includegraphics[scale=0.25]{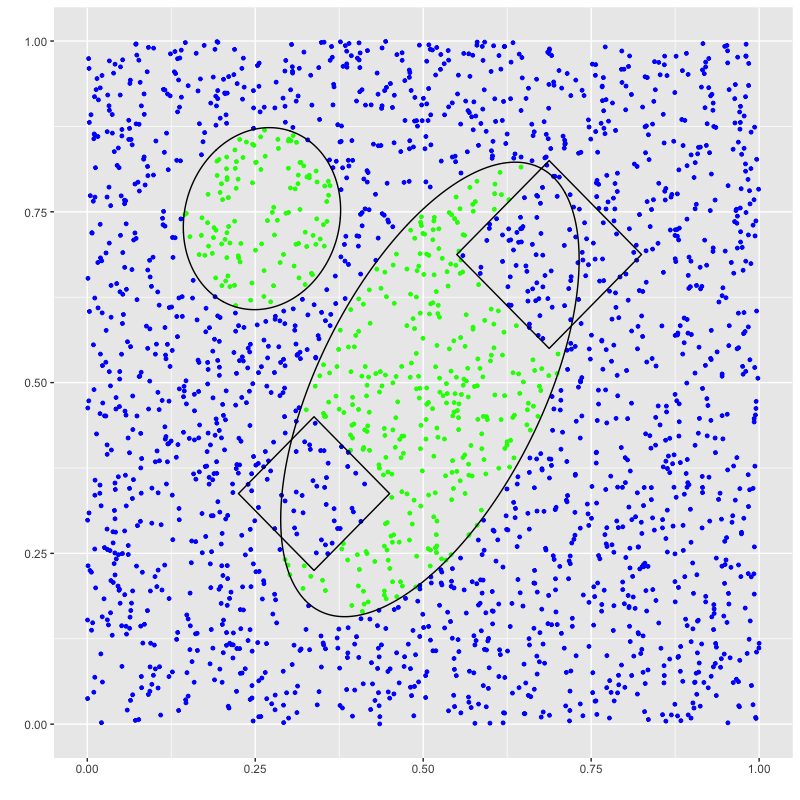}
    \includegraphics[scale=0.25]{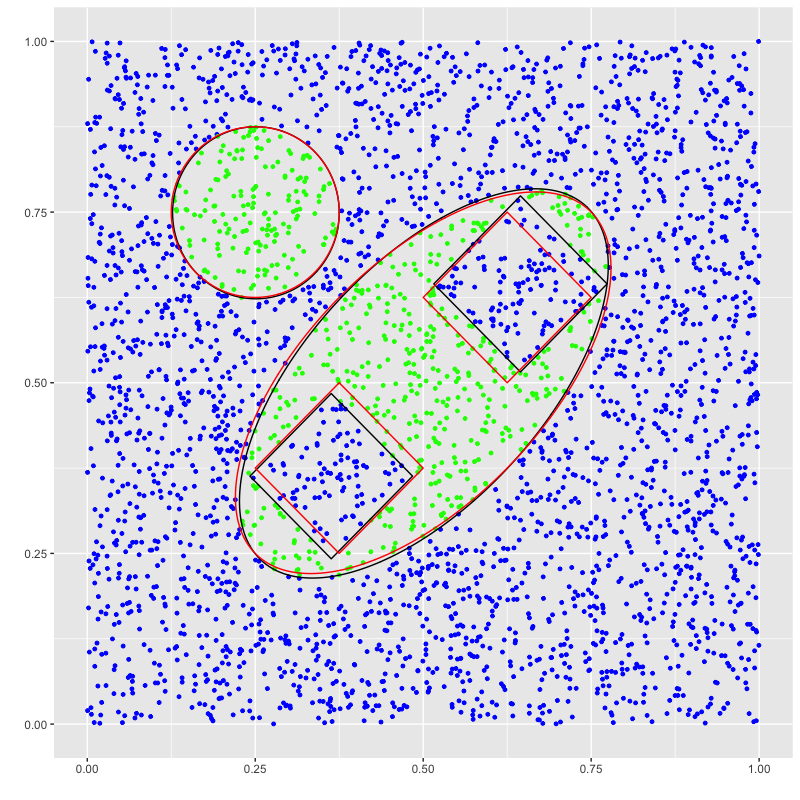}
    \caption{Sample sizes range from 20 to 5,000 and the corresponding grid sizes are between 1956 and 4259. By design, the estimated region (black) becomes very close to the true region (red) as sample size increases. Green grid points are determined based on the estimated region.}
    \label{sim1}
\end{figure}

In this first simulation study, we focus on parameter vectors of two dimensions for the purpose of visualization. We first make up some true confidence region of interest and a consistent estimate of the true region, followed by generating a grid which we evaluate using the consistent region. Then, we fit the SVM classifier using the labeled grid and compare its predictions with the true region. We also let the cardinality of the grid grow with sample size.

\begin{figure}[b!]
    \centering
    \includegraphics[scale=0.25]{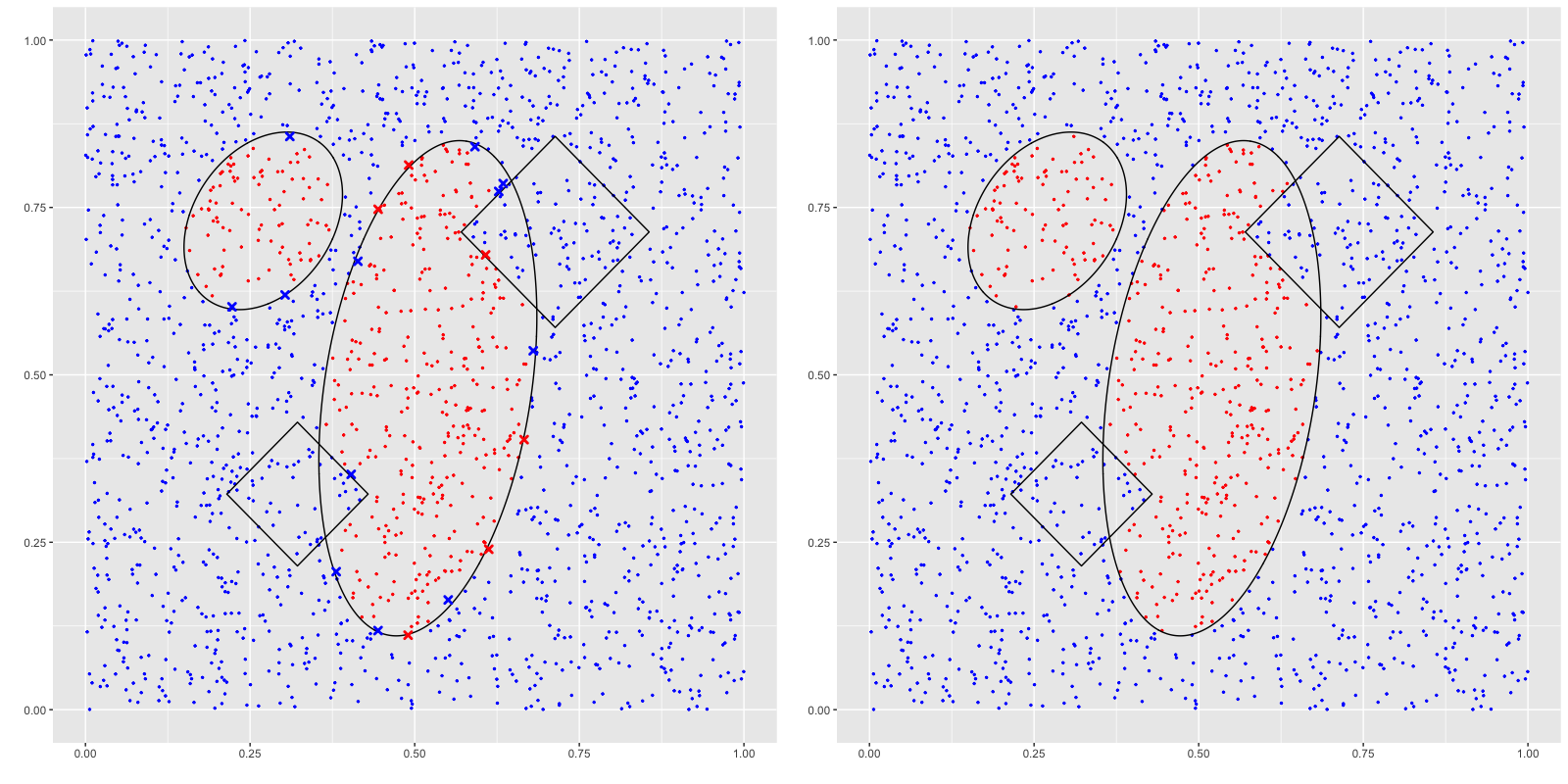}
    \includegraphics[scale=0.25]{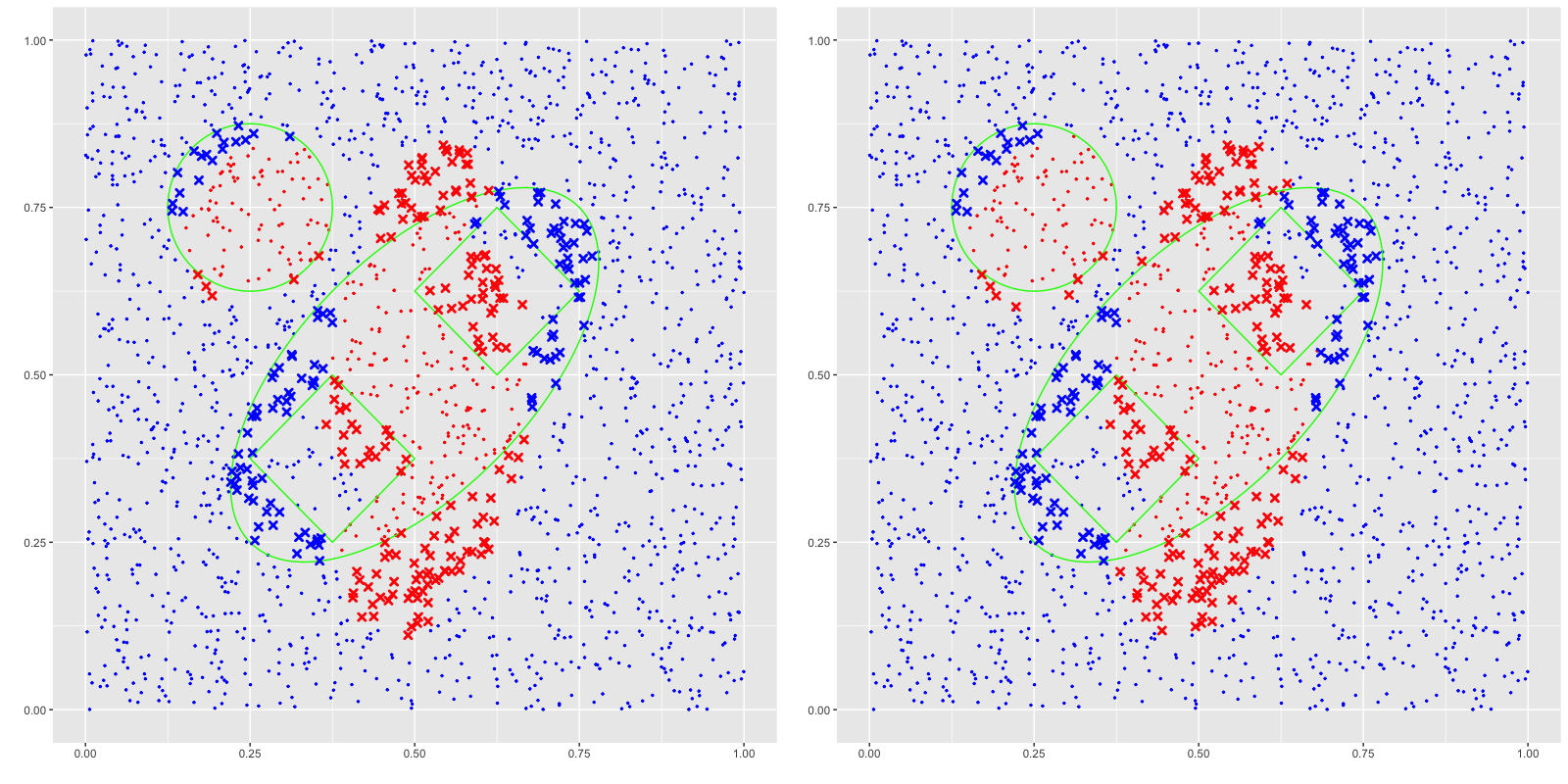}
    \caption{Sample size $n=50$. Left panels correspond to the same tuning and the right panels are governed by another set of tuning. Perfect training classification is achieved in the top-right plot; however, this does not shield us from making misclassification errors due to the small sample size.}
    \label{sim1_50}
\end{figure}

To begin with, let the true region be given by the union of (i) the red circle and (ii) the red ellipse minus the two diamonds inside the ellipse as showed in the top-left and bottom-right panels of Figure \ref{sim1}. Pretending that we have no prior information, we use a non-informative grid generated using Monte Carlo sequences whose cardinality is a function of the sample size. Particularly, 
$$|S|=\bigg[500*\log(n)\bigg],$$
where $\bigg[\cdot\bigg]$ rounds the number inside to the nearest integer.  For demonstration purposes, the consistent estimate of the true region obtained from our procedure, color-coded black in Figure \ref{sim1}, is constructed by adding some noise to the true region which vanishes as the sample size increases, hence the estimate becomes more accurate with larger sample sizes as we see in Figure \ref{sim1}. We plug each grid point into the analytical expression of the consistent estimate to assign a label to it. Subsequently, we fit an SVM classifier with the coordinates of the grid points being the features and the labels being the classes of the observations, to phrase it in the machine learning terms. And lastly, we compare the labels of the grid points predicted by the trained SVM classifier with the labels dictated by the true region and mark the mistakes in the plots. We repeat the exercise described here for different sample sizes. The results when sample sizes are 50 and 5,000 are presented in Figure \ref{sim1_50} and Figure \ref{sim1_5000}, respectively.

\begin{figure}[b!]
    \centering
    \includegraphics[scale=0.25]{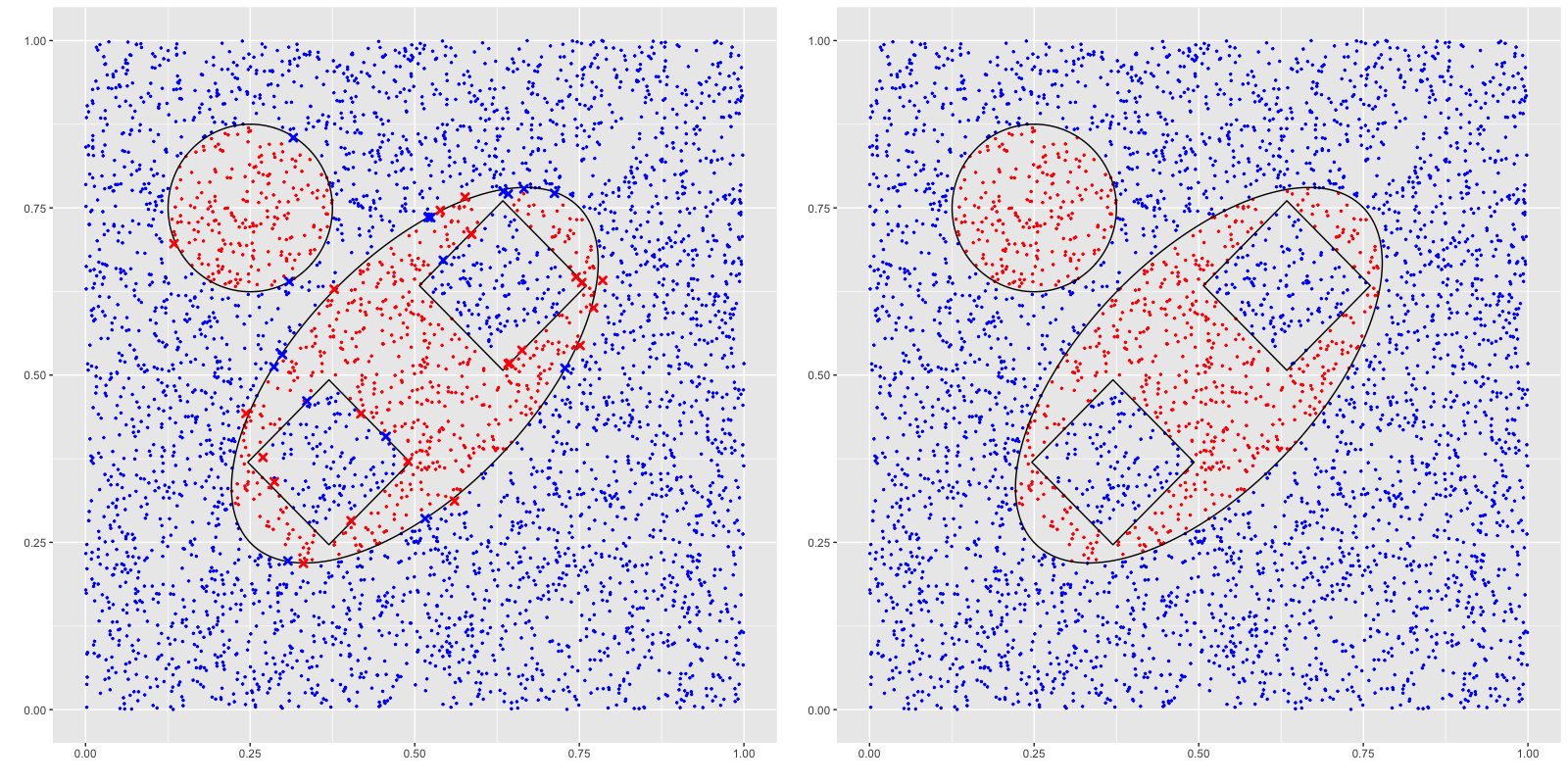}
    \includegraphics[scale=0.25]{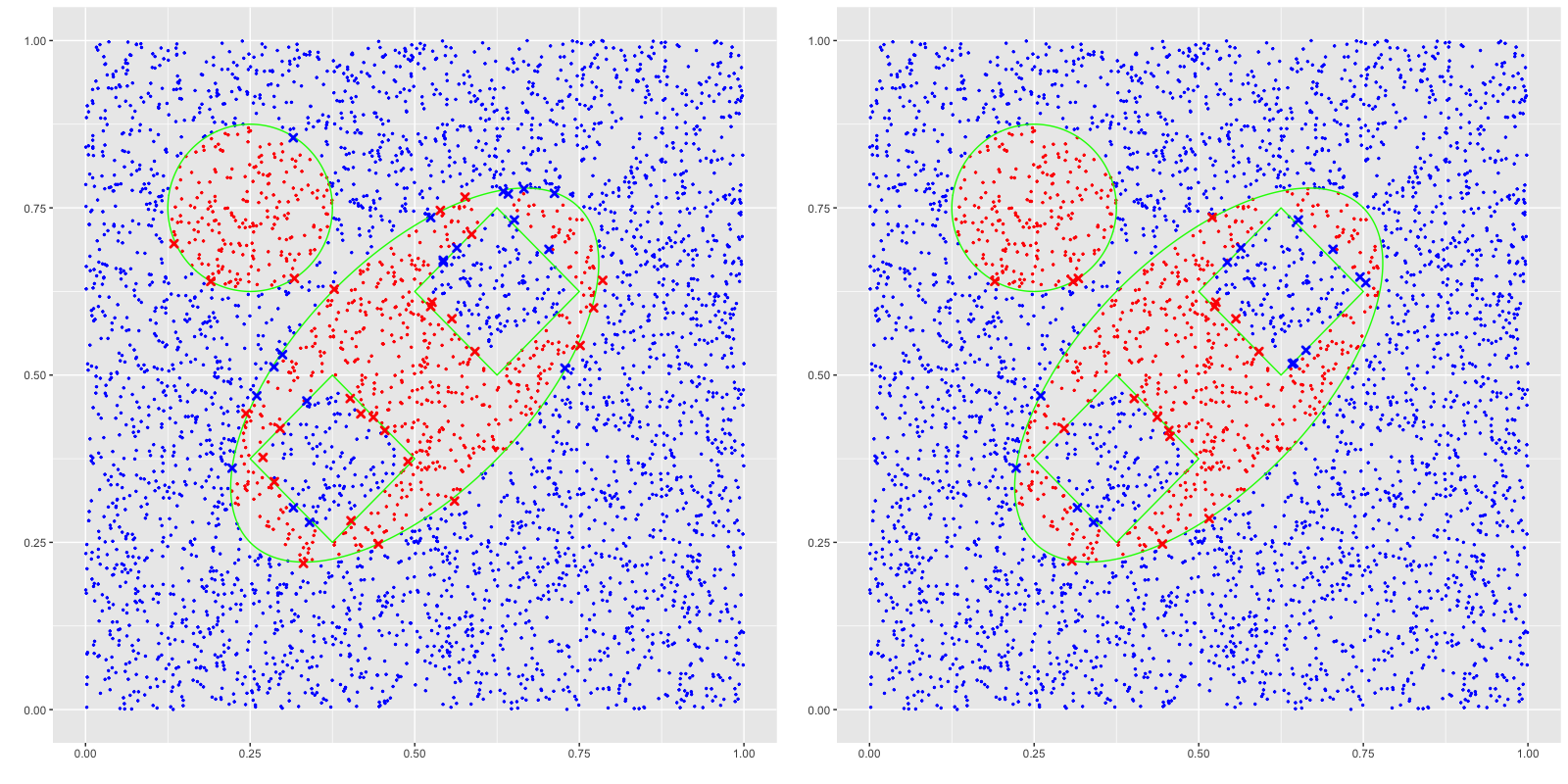}
    \caption{Sample size $n=5000$. Note that tuning for the right panels attains perfect classification (see top-right plot). For larger sample sizes and denser grids, achieving high training accuracy is more likely to make SVM classifier reproduce the true region.}
    \label{sim1_5000}
\end{figure}

The SVM classifiers are trained under different sets of tuning thus display different classification performance on the training grids. In the top panels, the black regions denote the consistent estimate of the true region which we use to label the grid points, red dots are the grid points predicted by the trained SVM classifiers to be ``inside'' the confidence set, and blue dots are predicted to be ``outside''. Furthermore, the crosses denote the mistakes by the SVM classifier. Specifically, red crosses are outside grid points that are predicted to be inside the confidence regions, and the blue crosses are inside grid points that are mistakenly predicted to be outside the region. The green regions in the bottom panels are the true confidence region. The predictions are similarly labeled when assessed against the true region in the bottom panels. We see that although we can easily achieve perfect classification on the training grid as is the case in the top-right panel of Figure \ref{sim1_50}, the bias due to the small sample size still results in tons of misclassification underneath the surface with respect to the true confidence region. Nonetheless, this procedure proves to be valid in that when we repeat the same steps under a sample size of 5,000, not only can we attain perfection on the training grid, the SVM classifier also largely captures the underlying true region as a composite result of the finer grid and the consistency of the estimate.

\subsection{OLS Coefficients}
In the second simulation study, we introduce a simple parametric OLS model and explore the out-of-sample performance of the SVM classifier. Specifically, we simulate data, including the covariates and shocks, based on the model we make up, and generate grid points which we label using the analytical formula for the confidence sets of OLS coefficients and which we use to train the SVM classifier. We examine the correct prediction rates over data that our classifier has not seen and record the computation time. The detailed steps are described as follows.

Consider a OLS model with an intercept and four covariates, so the parameter of interest is 5-dimensional provided that everything else is known and not of interest. So, we draw a random vector with five numbers taken as the true coefficients, $\beta_0$. With a multivariate design matrix and some error of length $n=500$, it is easy to generate the dependent variable, hence the estimate of $\beta_0$ as well as its covariance matrix. Next, we construct the Monte Carlo grid of cardinality $|S|=\bigg[6^{\log(n)}\bigg]=68,534$ and assign the labels with the following criterion. For each grid point $s$, we assign label $1$ if 
$$(s-\hat{\beta})^\top\hat{V}^{-1}(s-\hat{\beta})<\chi^2_5(0.95),$$
where $\hat{\beta}$ is the usual estimator for $\beta_0$, $\hat{V}$ is its covariance matrix, and $\chi^2_5(0.95)$ denotes the $95$th percentile of the $\chi^2$ distribution with $5$ degrees of freedom. Since a grid of such cardinality is very sparse in the $5$-dimensional real space, we end up with only $853$ grid points with label $1$, which is about $1.2\%$ of all grid points. For different ratios between the two sub-grids, we repeat the following $2,000$ times and compute the average outcomes. We split the grid into training grid and test grid. Then, we fit an SVM classifier using the training grid and their corresponding labels, and make it predict the labels of the test grid to see its prediction accuracy. Additionally, we also test whether $\beta_0$ is predicted to have label $1$, i.e. covered by the regions reproduced by the SVM classifier. Note that $\beta_0$ satisfies the above inequality, thus lies in the sample $95\%$ confidence set. For the purpose of comparison, we hold the tuning fixed throughout this subsection. The total run-time of $2,000$ iterations is reported for each configuration (i.e. training-test split ratio). Results are summarized in Table \ref{table_ols}. 

\begin{table}[H]
    \centering
    \begin{tabular}{ccccc}
        \hline
        Training & Test & Test accuracy (\%) & \% Capturing $\beta_0$ & Time (sec) \\
        \hline
        \hline
        54,827 & 13,707 & 99.15 & 100.00 & 23,630 \\
        28,784 & 39,750 & 98.99 & 94.65 & 8,146 \\
        13,707 & 54,827 & 98.86 & 45.65 & 3,120 \\
        3,427 & 65,107 & 98.77 & 8.8 & 1,140 \\
        \hline
    \end{tabular}
    \caption{Results from 2000 iterations of the described simulation study over a grid of $68,534$ points on the Roar cluster. }
    \label{table_ols}
\end{table}

Overall, both the prediction accuracy and coverage on $\beta_0$ grow when the classifier is trained with more data, i.e. grid points, ceteris paribus. Moreover, it only takes small training sizes to yield high test accuracies. On the other hand though, since incidentally very few grid points lie in the vicinity of $\beta_0$ in the original grid, coverages are low  when the training-test ratio is low. This should raise no concern because as the grid size grows in an equidistributed way, gaps are gradually filled with grid points. We see that the coverage of $\beta_0$ very quickly picks up as we train the SVM classifier with more grid points.  While run-time appears to grow exponentially with training data, the testing process takes almost no additional time. The implication of this in practice is, assigning labels on the denser grid using the trained SVM classifier will be very fast.

\subsection{Moment Condition Models}\label{simulation_cht}
We simulate the construction of the confidence set based on moment conditions in this subsection. In particular, we follow the approach proposed in \cite{CHT2007} (henceforth CHT) and employed in \cite{Ciliberto_Tamer} where the confidence set is based on keeping the all parameter vectors that satisfy an objective function resulted from moment conditions. That is, the collection of parameter values the evaluation of the criterion at which meets the moment conditions, i.e. 
\begin{equation}\label{cht_cs}
C_n(c)=\{\theta\in\Theta:n[\hat{Q}_n(\theta) - \min_t\hat{Q}(t)]\leq c\}
\end{equation}
Note that this approach involves minimization of the objective function for the entire data followed by minimization for each bootstrap subsample using some off-the-shelf solver, such as Nealder-Mead and other genetic algorithm like simulated annealing in \cite{Ciliberto_Tamer}. Subsequently, a grid search over $C_n(c_{(0)})$ given some initial $c_{(0)}$ is performed to find the maximum of $b(Q_b(\theta)-\min_t Q_b(t))$ for each bootstrap sample of size $b$. Once all bootstrap optimal values are computed, the empirical distribution of the values is formed and the empirical $\alpha$-quantile is chosen to replace $c$ in \ref{cht_cs}. \cite{Ciliberto_Tamer} repeats the above construction of critical value $c$ one more time, and then retains all qualified parameter vectors via another grid search. 

Consider a standard least squares problem with 3 parameter of interest, an intercept and two slope coefficients, as the objective. We set the sample size $n=1,000$ and bootstrap sample size $b=250$, and draw some $X\in\mathbb{R}^3$ to make $Y$ for each observation. For simplicity, we assume the model is correctly specified and hence obviating the inner minimization in the the above construct (\ref{cht_cs}), which reduces to the set $\{\theta\in \Theta: n\hat{Q}_n(\theta) \leq c\}$ where $\hat{Q}_n(\theta)=\frac{1}{n}\sum_i(Y_i-X_i^\top\theta)^2$. We use Nealder-Mead solver to obtain the minimal values of $\hat{Q}_n(\theta)$ using the entire sample as well as in the bootstrap sub-routine. The parameter space, a subset of $\mathbb{R}^3$, is divided into a uniform grid on the cube consisting of 101 values for each parameter and hence a total of $101^3$ grid points. This is arguably a very coarse grid since the points along each parameter dimension are fairly sparsely located, containing only 5 values between each pair of integers. Nonetheless, performing the brute-force CHT routine on such grid, with a programming choice of R without parallelization, runs for 3 hours on Penn State Roar Collab server with 16GB memory. We repeat the same exercise on a smaller grid with 51 points along each dimension and use the larger grid as benchmark to examine the difference between standard procedure and the use of SVM classifier. 

\begin{figure}
    \centering
    \includegraphics[scale=0.15]{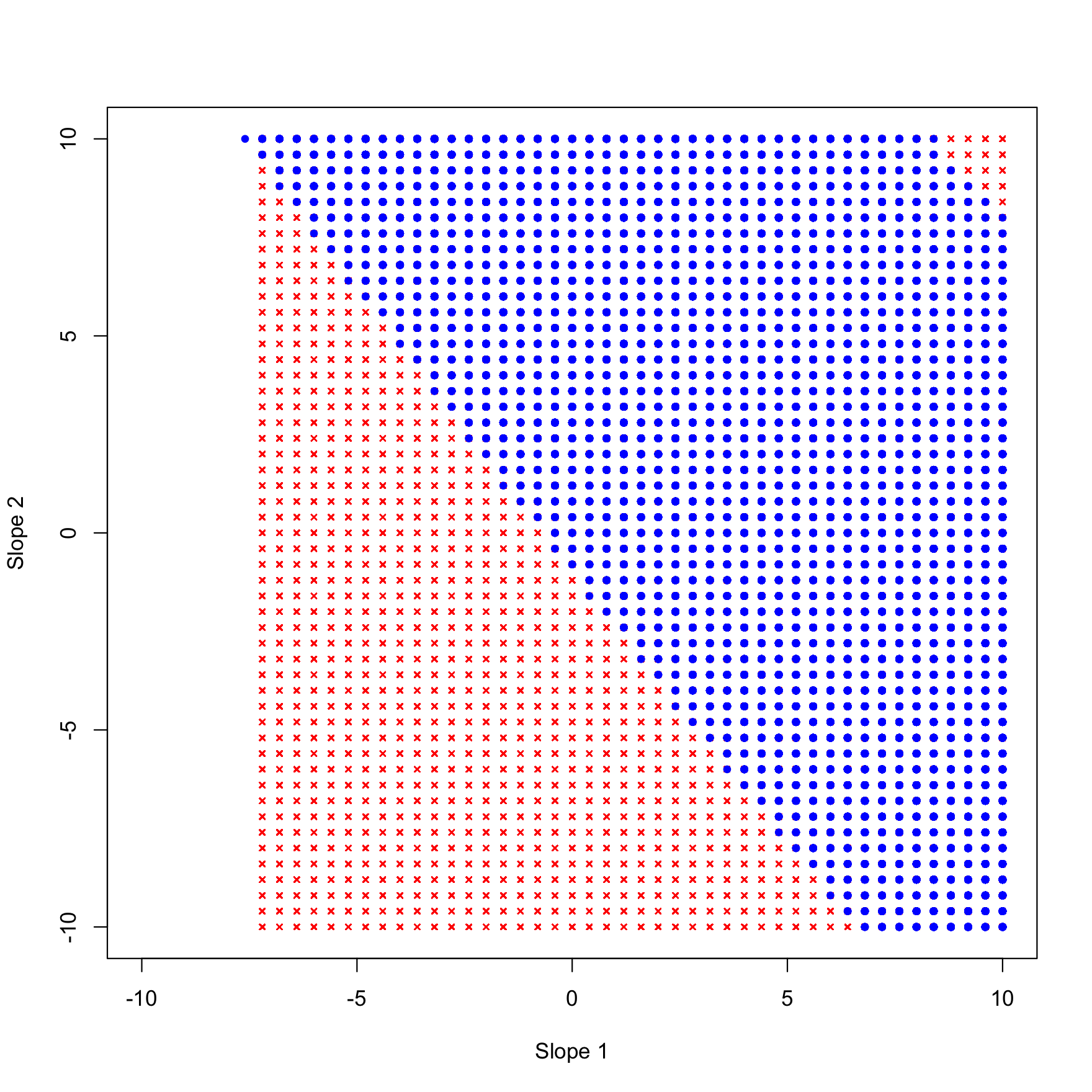}
    \caption{The accepted grid points (\textcolor{blue}{blue}) form a non-rectangular area. Reporting interval for each parameter leads readers to potentially falsely accept points outside of the joint confidence set (\textcolor{red}{red}).}
    \label{cht_cs}
\end{figure}

The relationship between the two slope coefficients is depicted in Figure \ref{cht_cs} based on the search on the $51^3$ grid. The blue dots are accepted according to CHT procedure and the red crosses are rejected. We note that once we report the confidence interval for each parameter, information is lost in that both blue dots and red crosses appear to be in the joint rectangle. For the SVM classifier, we train the model based on the coarser grid and examine its ability to capture the underlying confidence region using the large grid of $101^3$ points. As described above, labels for the purpose of classification are assigned according to each grid point's membership to the confidence set, i.e. 1 if the grid point is accepted according to CHT and -1 otherwise. We calculate the percentages of grid points classified to be in the CHT CS out of all grid points on both the training grid (i.e. the $51^3$ grid) and the test grid (i.e. the $101^3$ grid), denoted as ``\% Labeled CS''. We also compare the predicted labels against the labels dictated by the CHT procedure, and denote the percentage of labels agreed by both approaches as ``\% Correct''. Results are reported in Table \ref{table_cht}. The same calculations for the CHT procedure with rectangular region is also included for exposition. We emphasize again that, rather than making wrong predictions, the collection of individual intervals does not tell us anything about the confidence set but its minimum bounding box, which the SVM classifier is able to overcome. 

An immediate observation is the rectangle falsely includes many grid points that do not belong to the CS and leads to lower statistical power due to the fact that the joint confidence region does not admit a rectangular shape. To correct this means we would have to store all the accepted grid points; and even so, there is generally no prescription on new parameter vectors, e.g. a point that lies in the convex hull but outside of the concave hull of the accepted grid points, without evaluating the criterion that is based on moment conditions. Therefore, it would be much more concise and precise to consider instead the decision function of the form in equation (\ref{decision_func}) for $K(\theta_1,\theta_2)=\exp\{||\theta_1-\theta_2||^2/\sigma^2\}$ with a properly tuned $\sigma^2$. This allows us to easily determine if any parameter vector including those outside the grid points belongs to the joint confidence region by simply evaluating the function at the parameter vector. Lastly, we note that while such grid of $51^3$ points is fairly coarse, it is already so abundant that tuning of the kernel is essentially unnecessary if time is of utmost important. In the un-tuned classifier, the number of support vectors is 4,901, which amount to summing up 4,901 evaluations of simple operations. It takes about 16 seconds in the aforementioned computing environment for the trained classifier to classify each 10,000 grid points. The properly tuned classifier is comprised of 56,004 support vectors in this case. 

\begin{table}
    \centering
    \begin{tabular}{ccccc}
        \hline
        Method & \% Labeled CS ($51^3$) & \% Correct ($51^3$) & \% Labeled CS ($101^3$) & \% Correct ($101^3$) \\
        \hline
        \hline
        Rectangle & 86.275 & 58.758 & 87.128 & 44.695 \\
        SVM (no tuning) & 45.053 & 99.928 & 45.165 & 98.540 \\
        SVM (tuned) & 45.032 & 100.000 & 44.681 & 99.024 \\
        \hline
    \end{tabular}
    \caption{Comparison between standard reporting of intervals and SVM classifiers.}
    \label{table_cht}
\end{table}

\section{Conclusion}\label{conclusion}
In this paper, we devise an efficient and reliable way of reproducing the confidence sets in the context of moment conditions. These problems typically involve inverting the test and sub-sampling or re-sampling when constructing critical values, and are computationally complex especially when the dimensionalily of the parameter grows. Combining the uses of the support vector machine and the equidistributed sequences, we are able to consistently re-produce the limiting confidence regions as the sample size grows under common assumptions and mild conditions. This procedure is very easy to implement in practice. Researchers can generate grid points and fit the SVM classifier utilizing the existing packages in standard programming languages. Researchers can also incorporate an adaptive procedure into the generation of grid points for better approximation of the confidence sets when the sample size is fixed. Additionally, the use of decision functions provides much higher efficiency and precision than reporting the rectangular regions from the exhaustive search. This approach might also be used to compute and characterize the identified sets, or the confidence sets of the identified sets, so long as they have a similar formulation.

\bibliographystyle{apalike}
\bibliography{lib.bib}

\begin{appendices}
\section{Equidistributed Sequences}\label{grid}
The equidistributed sequences are designed to best explore the spaces of our interest. Some sequences are random in nature, such as the Monte Carlo sequences, whereas some others display patterns as we will see below. We first introduce the definition of equidistributed sequences, and then provide an argument that they leave "no gap" in the space. In other words, for any point in the $d$-dimensional Euclidean space, we can approach the point arbitrarily closely with equidistributed sequences. Denote $\lambda_d$ as the $d$-dimensional Lebesgue measure. The definition of equidistributed sequences is given below, extending the 1-dimensional definition introduced in \cite{Chandrasekharan1969} to potentially high-dimensional space $\mathbb{R}^d$. 

\begin{definition}
A sequence $S=(s_1,s_2,s_3,\cdots)$ in $\mathbb{R}^d$ is said to be \textit{equidistributed} on the box $\Gamma\subset\mathbb{R}^d$ if for every box $\Delta\subseteq\Gamma$, 
$$\lim_{n\to\infty}\frac{\sum_{i=1}^n \mathbbm{1}\{s_i\in \Delta\}}{n}=\frac{\lambda_d(\Delta)}{\lambda_d(\Gamma)}.$$
\end{definition}

With this, we can proceed to show that grids generated using equidistributed sequences explore the space in an exhaustive manner. Specifically, for any point in the space, we will show that, in the limit, there will always exist a point from an equidistributed sequence which is in any arbitrarily small neighborhood of the point.

\begin{lemma}\label{lemma3}
Let $S=(s_1,s_2,s_3,\cdots)$ be an equidistributed sequence in $\Gamma\subset\mathbb{R}^d$. Then, $\forall\theta\in\Gamma$ and $\forall\epsilon>0$, there exists some $s^*\in S$ such that $||\theta-s^*||\leq\epsilon$.
\end{lemma}

\begin{proof}
For some $\theta\in\Gamma$, let $B_\epsilon(\theta)$ be the ball centered at $\theta$ with radius $\epsilon$. First, we consider $\theta\in\text{int}(\Gamma)$, and assume for now $B_\epsilon(\theta)\subset\Gamma$. Let $C_\epsilon(\theta)$ be the hypercube such that the distance from any vertex to its center $\theta$ is $\epsilon$. In 2 dimensions, $C_\epsilon(\theta)$ is the square centered at $\theta$ whose vertices lie on the circle $B_\epsilon(\theta)$; and in 3 dimensions, $C_\epsilon(\theta)$ is a cube (whose edges are of equal length) centered at $\theta$ with all vertices lying on the ball $B_\epsilon(\theta)$. Since $C_\epsilon(\theta)$ is a box in $\Gamma$ and $\lambda_d(C_\epsilon(\theta))>0$, it follows that $\frac{\lambda_d(C_\epsilon(\theta))}{\lambda_d(\Gamma)}>0$. Therefore, by definition of equidistributed sequences, we have that 
$$\frac{\lambda_d(C_\epsilon(\theta))}{\lambda_d(\Gamma)}=\lim_{n\to\infty}\frac{\sum_{i=1}^n \mathbbm{1}\{s_i\in C_\epsilon\}}{n}>0.$$
This implies that $S\bigcap C_\epsilon(\theta)\neq\emptyset$, because there must be a positive number of points in $S$ that lie in $C_\epsilon(\theta)$. Take $s^*\in(S\bigcap C_\epsilon(\theta))$, and we have $s^*\in B_\epsilon(\theta)$, or equivalently $||\theta-s^*||\leq\epsilon$. 

To complete the proof, if $B_\epsilon(\theta)\subsetneq\Gamma$, we take instead $\Bar{\epsilon}\in(0,\epsilon)$ such that $B_{\bar{\epsilon}}(\theta)\subset(B_\epsilon(\theta)\bigcap\Gamma)$ (which cannot be empty since $\theta\in\text{int}(\Gamma)$), and apply the above analysis to $B_{\bar{\epsilon}}(\theta)$, hence the point $s^*\in B_{\bar{\epsilon}}(\theta)\subset B_\epsilon(\theta)$. And lastly, if $\theta$ lies on a face (or edge) of $\Gamma$, we construct a hyperrectangle inside the intersection of $B_\epsilon(\theta)$ and $\Gamma$ with positive volume (with respect to $\lambda_d$) such that $\theta$ lies on the face (or edge) of the hyperrectangle. Then we apply the same argument as above.
\end{proof}

Now we present how the grid points look under different generating sequences. Figures \ref{fig:mc} through \ref{fig:baker} show in the two-dimensional space the grid points generated from Monte Carlo sequences, Sobol sequences, Weyl sequences, Bakers sequences, respectively. Going from the top-left panel to the bottom-right, we present the grids of 20, 100, 500, and 2000 grid points in each of the following four figures. 
\begin{figure}[H]
    \centering
    \includegraphics[scale=0.65]{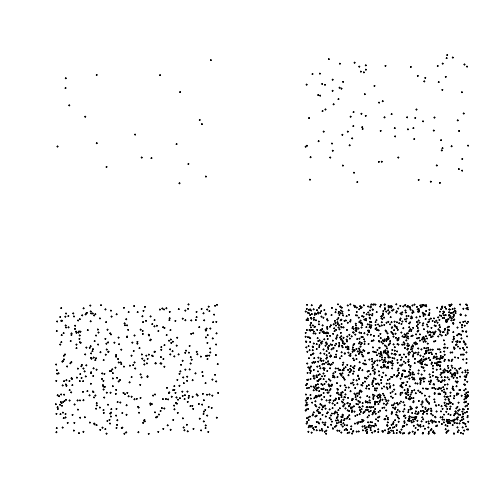}
    \caption{Monte Carlo Points}
    \label{fig:mc}
\end{figure}

\begin{figure}
    \centering
    \includegraphics[scale=0.65]{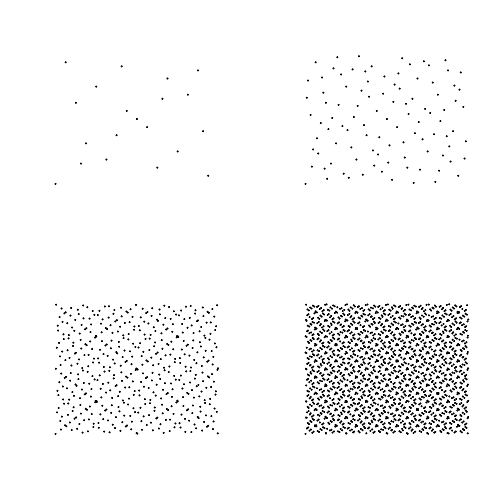}
    \caption{Sobol Points}
    \label{fig:sobol}
\end{figure}

\begin{figure}
    \centering
    \includegraphics[scale=0.65]{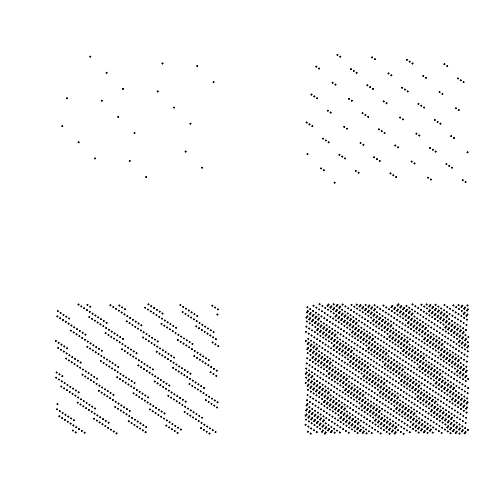}
    \caption{Weyl Points}
    \label{fig:weyl}
\end{figure}

\begin{figure}
    \centering
    \includegraphics[scale=0.65]{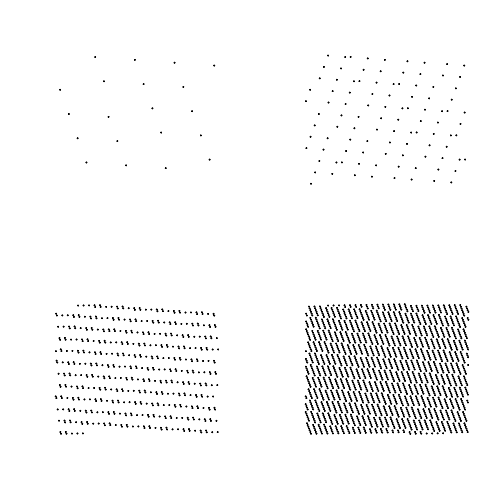}
    \caption{Baker Points}
    \label{fig:baker}
\end{figure}

We see that the Monte Carlo random sequences, although random in nature, leave out a lot of blank regions, hence might not buy us the optimal performance sampling the space. The uniformity issue is addressed better by the Sobol sequences. Besides the pattern, we see that the grid points are pretty evenly spread out in the space. Furthermore, Weyl and Baker points display patterns which can be use when we believe certain regions of the parameter space are under-visited. 

\section{Proof of the Main Results}\label{proofs}
\subsection{Proof of Lemma \ref{lemma2}}
\begin{proof}
We will use a standard argument based on the definition of the interior points. By definition, for every $\theta\in CS$, there exists a neighborhood $B(\theta)$ (with positive radius $\Bar{\epsilon}>0$) such that $B(\theta)\subset CS$. Since for every $n$, $E_n(\theta)\notin CS_n$, in the limit, it follows that $\underset{n\to\infty}{\lim}E_n(\theta)\notin CS$, thus the nearest exterior point will be at least $\Bar{\epsilon}$ away from $\theta$. We can simply set $M=\Bar{\epsilon}$. \\
Moreover, by Assumption \ref{ass:grid} and Lemma \ref{lemma3}, the grid will grow in a way that gradually fills the entire parameter space, i.e. there are grid points arbitrarily close to any point in the space. Then, it follows that 
$$||I_{n+1}-\theta||\leq ||I_n-\theta||$$
where 
$$I_n\equiv\underset{I\in S_n^{interior}}{\arg\min}\,||I-\theta||,$$ 
and $S_n^{interior}$ denotes those grid points of $S_n$ that are inside $CS$, because $I_n\in S_{n+1}$. We evaluate the growing grid using the same test given in $CS$ to avoid the issue of evolving boundaries, which can cause certain $\theta$ to be inside the $CS_n$ but outside $CS_{n'}$ for some $n\neq n'$. So, we have that
$$\lim_{n\to\infty}||I_n-\theta||=0.$$
\end{proof}

\subsection{Proof of Theorem \ref{thm1}}
\begin{proof}
The idea of this proof is to find a sufficient condition which implies the dominant influence of the nearest interior grid point over all exterior grid points by the construction of the grid points. Here we will assume that $\theta$ lies exactly on the boundary of $CS$ w.p.0, so we can focus on the case where $\theta\in\text{int}(CS)$. Since we want the following to be true 
$$\max_{i\in S_{n}^{interior}} K(i,\theta)>\sum_{j\in S_{n}^{exterior}}K(j,\theta),$$
where $S_n^{exterior}$ denotes those grid points of $S_n$ that are outside $CS$, if we can show the following sufficient condition, 
$$\max_{i\in S_{n}^{interior}} K(i,\theta)>|S_n^{exterior}|\cdot\max_{j\in S_{n}^{exterior}}K(j,\theta),$$
where $|\cdot|$ denotes the cardinality of the collection of grid points, we would establish that the influence of the nearest interior grid point dominates the combined influence of all exterior grid points, which then implies that our SVM classifier labels this point $\theta$ as 1. Equivalently, we would like to show the following condition holds true
$$\max_{i\in S_{n}^{interior}}\exp\left\{-\frac{||i-\theta||^2}{2\sigma_n^2}\right\} > |S_n^{exterior}|\cdot\max_{j\in S_{n}^{exterior}}\exp\left\{-\frac{||j-\theta||^2}{2\sigma_n^2}\right\}$$
$$\Leftrightarrow \exp\left\{-\frac{||I_n(\theta)-\theta||^2}{2\sigma_n^2}\right\}>|S_n^{exterior}|\cdot\exp\left\{-\frac{||E_n(\theta)-\theta||^2}{2\sigma_n^2}\right\},$$
which requires that 
$$\exp\left\{\frac{||E_n(\theta)-\theta||^2-||I_n(\theta)-\theta||^2}{2\sigma_n^2}\right\}>|S_n^{exterior}|,$$
or 
\begin{equation}\label{ineq:thm1}
0<2\sigma_n^2<\frac{||E_n(\theta)-\theta||^2-||I_n(\theta)-\theta||^2}{\log|S_n^{exterior}|}.
\end{equation}
By Lemma \ref{lemma2}, $\forall\theta\in \text{int}(CS)$, $\exists N\in\mathbb{N}$ such that $\forall n>N$, $||I_n(\theta)-\theta||^2<||E_n(\theta)-\theta||^2$, and hence such $\sigma_n^2$ exists that satisfies \ref{ineq:thm1}. Therefore, choosing the tuning parameter this way for each $n>N$ makes the influence of the nearest interior grid dominate the influence of all exterior grid points in the limit. Consequently, the SVM classifier trained over these grids under such tuning classifies any $\theta\in CS$ to have label 1; that is, $\{\theta\in CS\} \Rightarrow \{\theta\in SVM_+\}$. Because $\mathbb{P}(\theta_0\in CS)\geq 1-\alpha$, it follows that
$$\mathbb{P}(\theta_0\in SVM_+)\geq \mathbb{P}(\theta_0\in CS)\geq 1-\alpha$$
\end{proof}

\subsection{Proof of Corollary \ref{corollary1}}
\begin{proof}
The forward direction of this result is showed in the proof in the above section of Theorem \ref{thm1}. It only remains to show the backward direction; that is, if a point $\theta$ is classified by the SVM classifier to have label 1, then it must be in $CS$, too. We will show the contrapositive statement of the backward direction holds true. 

Take a point $\Tilde{\theta}\notin CS$, or $\theta\in\overline{CS}\equiv\Theta\setminus CS$, we want to show that we can find values of the tuning parameter $\sigma^2$ such that SVM classifier classifies $\Tilde{\theta}$ to have label $-1$. Following the same argument in the proof of Lemma \ref{lemma2}, as we expand the grid, the Euclidean distance from $\Tilde{\theta}$ to any interior grid point of $CS$ (with label $1$) is bounded away from 0; whereas the distance from $\Tilde{\theta}$ to the nearest exterior grid point (with label $-1$) decreases to $0$ in the limit. Notice that the interior and exterior grid points are defined in reference to $CS$, the same way as above. Mathematically, we have
\begin{enumerate}
    \item $\exists \Tilde{M}>0$ such that $||I_n(\Tilde{\theta})-\Tilde{\theta}||\geq \Tilde{M}>0$, for all $n$; and 
    \item $\underset{n\to\infty}{\lim}||E_n(\Tilde{\theta})-\Tilde{\theta}||\to 0$.
\end{enumerate}
Now, in order for SVM to classify $\Tilde{\theta}$ as $-1$, a sufficient condition is 
$$\max_{s_i\in S_n^{exterior}} K(s_i,\Tilde{\theta})>|S_+|\cdot \max_{s_j\in S_n^{interior}} K(s_j,\Tilde{\theta})$$
$$\Leftrightarrow \max_{s_i\in S_n^{exterior}} \exp\left(-\frac{||s_i-\Tilde{\theta}||^2}{2\sigma_n^2}\right)>|S_n^{interior}|\cdot \max_{s_j\in S_n^{interior}} \exp\left(-\frac{||s_j-\Tilde{\theta}||^2}{2\sigma_n^2}\right)$$
Solving the above inequality yields
\begin{equation}\label{ineq:corollary1}
0<2\sigma_n^2<\frac{||I_n(\Tilde{\theta})-\Tilde{\theta}||^2-||E_n(\Tilde{\theta})-\Tilde{\theta}||^2}{\log|S_n^{interior}|}
\end{equation}
Since for large enough sample size $n$, $||I_n(\Tilde{\theta})-\Tilde{\theta}||^2>||E_n(\Tilde{\theta})-\Tilde{\theta}||^2$, the chain of inequalities in \ref{ineq:corollary1} is well-defined. Therefore, we can find values for the tuning parameter $\sigma_n^2$ for all $n$'s that are sufficiently large such that a point $\Tilde{\theta}$ in the complement of $CS$ will be labelled $-1$ by the SVM classifier. \\
Combining this with Theorem \ref{thm1}, it follows that when the sample size is large enough, by taking the values in the intersection of \ref{ineq:thm1} and \ref{ineq:corollary1} for the tuning parameter $\sigma^2$, we have that a point 
$\theta\in\Theta$ belongs to $CS$ if and only if the SVM classifier under such tuning classifies it to have label $1$.
\end{proof}

\end{appendices}


\end{document}